\documentclass[12pt,cls,onecolumn]{IEEEtran}
\usepackage{graphicx,subfigure, amsmath, amssymb}
\begin{document}

\title{Improved Capacity Scaling in Wireless Networks With Infrastructure}

\author{\large Won-Yong Shin, \emph{Member}, \emph{IEEE},
Sang-Woon Jeon, \emph{Student Member}, \emph{IEEE},\\ Natasha
Devroye, Mai H. Vu, \emph{Member}, \emph{IEEE}, Sae-Young Chung,
\emph{Senior Member}, \emph{IEEE}, Yong H. Lee,
\emph{Senior Member}, \emph{IEEE}, and Vahid Tarokh, \emph{Fellow}, \emph{IEEE}\\
\thanks{The work of W.-Y. Shin and Y. H. Lee was supported by the Brain Korea 21
Project, The School of Information Technology, KAIST in 2008. The
work of S.-W. Jeon and S.-Y. Chung was supported by the MKE under
the ITRC support program supervised by the IITA
(IITA-2008-C1090-0803-0002). The work of M. H. Vu was supported in
part by ARO MURI grant number W911NF-07-1-0376. The material in this
paper was presented in part at the IEEE Communication Theory
Workshop, St. Croix, US Virgin Islands, May 2008 and the IEEE
International Symposium on Information Theory, Toronto, Canada, July
2008.}
\thanks{W.-Y. Shin was with the Department of EE, KAIST, Daejeon 305-701, Republic of Korea. He is now with the School of Engineering and Applied Sciences,
Harvard University, Cambridge, MA 02138 USA (E-mail:
wyshin@seas.harvard.edu).}
\thanks{S.-W. Jeon, S.-Y. Chung, and Y. H. Lee are with the Department of EE, KAIST, Daejeon 305-701, Republic of Korea (E-mail: swjeon@kaist.ac.kr; sychung@ee.kaist.ac.kr;
yohlee@ee.kaist.ac.kr).}
\thanks{N. Devroye was with the School of Engineering and Applied Sciences, Harvard
University, Cambridge, MA 02138 USA. She is now with the Department
of Electrical and Computer Engineering, University of Illinois at
Chicago, Chicago, Illinois 60607 USA (E-mail: devroye@ece.uic.edu).}
\thanks{M. H. Vu was with the School of Engineering and Applied Sciences, Harvard
University, Cambridge, MA 02138 USA. She is now with the Department
of Electrical and Computer Engineering, McGill University, Montreal,
QC H3A 2A7, Canada (E-mail: mai.h.vu@mcgill.ca).}
\thanks{V. Tarokh
is with the School of Engineering and Applied Sciences, Harvard
University, Cambridge, MA 02138 USA (E-mail:
vahid@seas.harvard.edu).} } \maketitle

\markboth{Under Revision for IEEE Transactions on Information
Theory} {Shin {\em et al.}: Improved Capacity Scaling in Wireless
Networks With Infrastructure}


\newtheorem{definition}{Definition}
\newtheorem{theorem}{Theorem}
\newtheorem{lemma}{Lemma}
\newtheorem{example}{Example}
\newtheorem{corollary}{Corollary}
\newtheorem{proposition}{Proposition}
\newtheorem{conjecture}{Conjecture}
\newtheorem{remark}{Remark}

\def \diag{\operatornamewithlimits{diag}}
\def \min{\operatornamewithlimits{min}}
\def \max{\operatornamewithlimits{max}}
\def \log{\operatorname{log}}
\def \max{\operatorname{max}}
\def \rank{\operatorname{rank}}
\def \out{\operatorname{out}}
\def \exp{\operatorname{exp}}
\def \arg{\operatorname{arg}}
\def \E{\operatorname{E}}
\def \tr{\operatorname{tr}}
\def \SNR{\operatorname{SNR}}
\def \dB{\operatorname{dB}}
\def \ln{\operatorname{ln}}

\def \be {\begin{eqnarray}}
\def \ee {\end{eqnarray}}
\def \ben {\begin{eqnarray*}}
\def \een {\end{eqnarray*}}

\newpage

\begin{abstract}

This paper analyzes the impact and benefits of infrastructure
support in improving the throughput scaling in networks of $n$
randomly located wireless nodes. The infrastructure uses
multi-antenna base stations (BSs), in which the number of BSs and
the number of antennas at each BS can scale at arbitrary rates
relative to $n$. Under the model, capacity scaling laws are analyzed
for both dense and extended networks. Two BS-based routing schemes
are first introduced in this study: an infrastructure-supported
single-hop (ISH) routing protocol with multiple-access uplink and
broadcast downlink and an infrastructure-supported multi-hop (IMH)
routing protocol. Then, their achievable throughput scalings are
analyzed. These schemes are compared against two conventional
schemes without BSs: the multi-hop (MH) transmission and
hierarchical cooperation (HC) schemes. It is shown that a linear
throughput scaling is achieved in dense networks, as in the case
without help of BSs. In contrast, the proposed BS-based routing
schemes can, under realistic network conditions, improve the
throughput scaling significantly in extended networks. The gain
comes from the following advantages of these BS-based protocols.
First, more nodes can transmit simultaneously in the proposed scheme
than in the MH scheme if the number of BSs and the number of
antennas are large enough. Second, by improving the long-distance
signal-to-noise ratio (SNR), the received signal power can be larger
than that of the HC, enabling a better throughput scaling under
extended networks. Furthermore, by deriving the corresponding
information-theoretic cut-set upper bounds, it is shown under
extended networks that a combination of four schemes IMH, ISH, MH,
and HC is order-optimal in all operating regimes.

\end{abstract}



\newpage

\section{Introduction}

In~\cite{GuptaKumar:00}, Gupta and Kumar introduced and studied the
throughput scaling in a large wireless ad hoc network. They showed
that, for a network of $n$ source--destination (S--D) pairs randomly
distributed in a unit area, the total throughput scales as
$\Theta(\sqrt{n/\log n})$.\footnote{We use the following notations:
i) $f(x)=O(g(x))$ means that there exist constants $C$ and $c$ such
that $f(x)\le Cg(x)$ for all $x>c$. ii) $f(x)=o(g(x))$ means that
$\underset{x\rightarrow\infty}\lim\frac{f(x)}{g(x)}=0$. iii)
$f(x)=\Omega(g(x))$ if $g(x)=O(f(x))$. iv) $f(x)=\omega(g(x))$ if
$g(x)=o(f(x))$. v) $f(x)=\Theta(g(x))$ if $f(x)=O(g(x))$ and
$g(x)=O(f(x))$~\cite{Knuth:76}.} This throughput scaling is achieved
using a multi-hop (MH) communication scheme. Recent results have
shown that an almost linear throughput in the network, i.e.,
$\Theta(n^{1-\epsilon})$ for an arbitrarily small $\epsilon>0$, is
achievable by using a hierarchical cooperation (HC)
strategy~\cite{OzgurLevequeTse:07,NiesenGuptaShah:09,Xie:08,NiesenGuptaShah:08}.
Besides the schemes
in~\cite{OzgurLevequeTse:07,NiesenGuptaShah:09,Xie:08,NiesenGuptaShah:08},
there has been a steady push to improve the throughput of wireless
networks up to a linear scaling in a variety of network scenarios by
using novel techniques such as networks with node
mobility~\cite{GrossglauserTse:02}, interference alignment
schemes~\cite{CadambeJafar:07}, and infrastructure
support~\cite{ZemlianovVeciana:05}.

Although it would be good to have such a linear scaling with only
wireless connectivity, in practice there will be a price to pay in
terms of higher delay and higher cost of channel estimation. For
these reasons, it would still be good to have infrastructure aiding
wireless nodes. Such hybrid networks consisting of both wireless ad
hoc nodes and infrastructure nodes, or equivalently base stations
(BSs), have been introduced and analyzed
in~\cite{KulkarniViswanath:03,KozatTassiulas:03,ZemlianovVeciana:05,LiuLiuTowsley:03,LiuThiranTowsley:07}.
BSs are assumed to be interconnected by high capacity wired links.

While it has been shown that BSs can be beneficial in wireless
networks, the impact and benefits of infrastructure support are not
fully understood yet. This paper features analysis of the capacity
scaling laws for a more general hybrid network where there are $l$
antennas at each BS, allowing the exploitation of the spatial
dimension at each BS.\footnote{When the carrier frequency is very
high, it is possible to deploy many antennas at each BS since the
wavelength becomes very small.} By allowing the number $m$ of BSs
and the number $l$ of antennas to scale at arbitrary rates relative
to the number $n$ of wireless nodes, achievable rate scalings and
information-theoretic upper bounds are derived as a function of
these scaling parameters. To show our achievability results, two new
routing protocols utilizing BSs are proposed here. In the first
protocol, multiple sources (nodes) transmit simultaneously to each
BS using a direct single-hop multiple-access in the uplink and a
direct single-hop broadcast from each BS in the downlink. In the
second protocol, the high-speed BS links are combined with
nearest-neighbor routing via MH among the wireless nodes. The
obtained results are also compared to two conventional schemes
without using BSs: the MH protocol~\cite{GuptaKumar:00} and HC
protocol~\cite{OzgurLevequeTse:07}.

The proposed schemes are evaluated in two different networks: dense
networks~\cite{GuptaKumar:00,ElGamalMammenPrabhakarShah:06,OzgurLevequeTse:07}
of unit area, and extended
networks~\cite{XieKumar:04,JovicicViswanathKulkarni:04,XueXieKumar:05,FranceschettiDouseTseThiran:07,OzgurLevequeTse:07}
of unit node density. In dense networks, it is shown that an almost
linear throughput scaling is achieved as
in~\cite{OzgurLevequeTse:07}, which is rather obvious. On the
contrary, in extended networks, depending on the network
configurations and the path-loss attenuation, the proposed BS-based
protocols can improve the throughput scaling significantly, compared
to the case without help of BSs. Part of the improvement comes from
the following two advantages over the conventional schemes: having
more antennas enables more transmit pairs that can be activated
simultaneously (compared to those of the MH scheme), i.e., enough
degree-of-freedom (DoF) gain is obtained, provided the number $m$ of
BSs and the number $l$ of antennas per BS are large enough. In
addition, the BSs help to improve the long-distance signal-to-noise
ratio (SNR)\footnote{In~\cite{OzgurJohariTseLeveque:10}, the
long-distance SNR is defined as $n$ times the received SNR between
two farthest nodes across the largest scale in wireless networks. In
our BS-based network, it can be interpreted as the total SNR
transferred to any given node (or BS antenna) over a certain smaller
scale reduced by infrastructure support.}, first termed
in~\cite{OzgurJohariTseLeveque:10}, which leads to a larger received
signal power than that of the HC scheme, i.e., enough power gain is
obtained, thus allowing for a better throughput scaling in extended
networks.

To show the optimality of our proposed schemes, cut-set upper bounds
on the throughput scaling are derived for a network with
infrastructure. Note that the previous upper bounds
in~\cite{XieKumar:04,JovicicViswanathKulkarni:04,LevequeTelatar:05,XueXieKumar:05,OzgurLevequePreissmann:07,OzgurLevequeTse:07}
assume pure ad hoc networks and those for BS-based networks are not
rigorously characterized in both dense and extended networks. In
dense networks, it is shown that the obtained upper bound is the
same as that of~\cite{OzgurLevequeTse:07} assuming no BSs. Hence, it
is seen that the BSs cannot improve the throughput scaling and the
HC scheme is order-optimal for all the operating regimes. In
extended networks, the proposed approach is based in part on the
characteristics at power-limited regimes shown
in~\cite{OzgurLevequeTse:07,OzgurJohariTseLeveque:10}. It is shown
that our upper bounds match the achievable throughput scalings for
all the operating regimes within a factor of $n^{\epsilon}$ with an
arbitrarily small exponent $\epsilon>0$. To achieve the
order-optimal scaling, using one of the two BS-based routings,
conventional MH transmission, and HC strategy is needed depending on
the operating regimes.

The rest of this paper is organized as follows. Section
\ref{SEC:system} describes the proposed network model with
infrastructure support. The main results are briefly shown in
Section~\ref{SEC:result}. The two proposed BS-based protocols are
characterized in Section \ref{SEC:routing} and their achievable
throughput scalings are analyzed in Section \ref{SEC:scaling}. The
corresponding information-theoretic cut-set upper bounds are derived
in Section \ref{SEC:upper}. Finally, Section \ref{SEC:conclusion}
summarizes this paper with some concluding remarks.

Throughout this paper the superscripts $T$ and $\dagger$ denote the
transpose and conjugate transpose, respectively, of a matrix (or a
vector). ${\bf I}_n$ is the identity matrix of size $n\times n$,
$[\cdot]_{ki}$ is the $(k, i)$-th element of a matrix, $\tr(\cdot)$
is the trace, and $\det(\cdot)$ is the determinant. $\mathbb{C}$ is
the field of complex numbers and $E[\cdot]$ is the expectation.
Unless otherwise stated, all logarithms are assumed to be to the
base 2.


\section{System and Channel Models} \label{SEC:system}

Consider a two-dimensional wireless network that consists of $n$
S--D pairs uniformly and independently distributed on a square
except for the area covered by BSs. Then, no nodes are physically
located inside the BSs. The network is assumed to have an area of
one and $n$ in dense and extended networks, respectively. Suppose
that the whole area is divided into $m$ square cells, each of which
is covered by one BS with $l$ antennas at its center (see Fig.
\ref{FIG:infra}). It is assumed that the total number of antennas in
the network scales at most linearly with $n$, i.e., $ml=O(n)$. For
analytical convenience, we would like to state that parameters $n$,
$m$, and $l$ are then related according to
\begin{equation}
  n=m^{1/\beta}=l^{1/\gamma}, \nonumber
\end{equation}
where $\beta, \gamma \in [0,1)$ satisfying $\beta+\gamma\le1$. The
number of antennas is allowed to grow with the number of nodes and
BSs in the network. The placement of these $l$ antennas depends on
how the number of antennas scales as follows:
\begin{enumerate}
  \item $l$ antennas are regularly placed on the BS boundary if
  $l=O(\sqrt{n/m})$, and
  \item $\sqrt{n/m}$ antennas are regularly placed on the BS boundary and the rest are uniformly placed inside the
    boundary if $l=\omega(\sqrt{n/m})$ and
    $l=O(n/m)$.\footnote{Such an antenna deployment strategy
guarantees both the nearest neighbor transmission from/to each BS
antenna and enough space among the antennas, and thus enables our
BS-based routing protocol via MH to work well, which will be
discussed in Section~\ref{SEC:withinfra}.}
\end{enumerate}
Furthermore, we assume that the BSs are connected to each other with
sufficiently large capacity such that the communication between the
BSs is not the limiting factor to overall throughput scaling. The
required transmission rate of each wired BS-to-BS link will be
specified later (in Remark~\ref{REM:CBS}). It is also assumed that
these BSs are neither sources nor destinations. Suppose that the
radius of each BS scales as $\epsilon_0/\sqrt{m}$ for dense networks
and as $\epsilon_0\sqrt{n/m}$ for extended networks, where
$\epsilon_0>0$ is an arbitrarily small constant independent of $n$,
which means that it is independent of $m$ and $l$ as well. This
radius scaling would ensure enough separation among the antennas
since the per-antenna distance scales at least as the average
per-node distance for any parameters $n$, $m$, and $l$.

Let us first describe the uplink channel model. Let
$I\subseteq\{1,\cdots,n\}$ denote the set of simultaneously
transmitting wireless nodes. Then, the $l\times1$ received signal
vector ${\bf y}_{s}$ of BS $s\in\{1,\cdots,m\}$ and the $l\times1$
uplink complex channel vector ${\bf h}_{si}^u$ between node
$i\in\{1,\cdots,n\}$ and BS $s$ are given by
\begin{equation}
{\bf y}_{s}=\underset{i\in I}\sum {\bf h}_{si}^{u}x_{i}+{\bf n}_{s}
\nonumber
\end{equation}
and
\begin{equation}
{\bf
h}_{si}^u=\left[\frac{e^{j\theta_{si,1}^u}}{r_{si,1}^{u~\alpha/2}} \
    \frac{e^{j\theta_{si,2}^u}}{r_{si,2}^{u~\alpha/2}} \cdots \
    \frac{e^{j\theta_{si,l}^u}}{r_{si,l}^{u~\alpha/2}}\right]^T,
    \label{EQ:signalU}
\end{equation}
respectively, where $x_{i}$ is the transmit signal of node $i$, and
${\bf n}_s$ denotes the circularly symmetric complex additive white
Gaussian noise (AWGN) vector whose element has zero-mean and
variance $N_0$. Here, ${\theta_{si,t}^u}$ represents the random
phases uniformly distributed over $[0,2\pi)$ and independent for
different $i$, $s$, $t$, and time (transmission symbol), i.e., fast
fading. Note that this random phase model is based on a far-field
assumption, which is valid if the wavelength is sufficiently small.
${r_{si,t}^{u}}$ and $\alpha>2$ denote the distance between node $i$
and the $t$-th antenna of BS $s$, and the path-loss exponent,
respectively. Similarly, the $1\times l$ downlink complex channel
vector ${\bf h}_{is}^d$ between BS $s$ and node $i$
($s\in\{1,\cdots,m\}$ and $i\in\{1,\cdots,n\}$) and the complex
channel $h_{ki}$ between nodes $i$ and $k$ ($i,k\in\{1,\cdots,n\}$)
are given by
\begin{equation}
{\bf
h}_{is}^d=\left[\frac{e^{j\theta_{is,1}^d}}{r_{is,1}^{d~\alpha/2}} \
    \frac{e^{j\theta_{is,2}^d}}{r_{is,2}^{d~\alpha/2}} \cdots \
    \frac{e^{j\theta_{is,l}^d}}{r_{is,l}^{d~\alpha/2}}\right] \label{EQ:signalD} \nonumber
\end{equation}
and
\begin{equation}
h_{ki}=\frac{e^{j\theta_{ki}}}{r_{ki}^{\alpha/2}},
\label{EQ:signalW}
\end{equation}
respectively, where ${\theta_{is,t}^d}$ and ${\theta_{ki}}$ have
uniform distribution over $[0,2\pi)$, and are independent for
different $i$, $s$, $t$, $k$, and time. ${r_{is,t}^{d}}$ and
$r_{ki}$ denote the distance between the $t$-th antenna of BS $s$
and node $i$, and the distance between nodes $i$ and $k$,
respectively.

Suppose that each node and BS should satisfy an average transmit
power constraint $P$ and $nP/m$, respectively, during
transmission.\footnote{This assumption is reasonable since the
balance between uplink and downlink would be maintained for the case
where the transmit power of one BS in a cell increases
proportionally to the total power consumed by all the users covered
by the cell. In this case, note that although we allow additional
power for BSs, the total transmit power used by all wireless nodes
and BSs still remains as $\Theta(n)$.} Then, the total transmit
powers allowed for the $n$ wireless nodes and the $m$ BSs are the
same. Channel state information (CSI) is assumed to be available
both at the receivers and the transmitters for downlink
transmissions from BSs but only at the receivers for transmissions
from wireless nodes. Let $T_n(\alpha,\beta,\gamma)$ denote the total
throughput of the network for the parameters $\alpha$, $\beta$, and
$\gamma$, and then its scaling exponent is defined
by~\cite{OzgurLevequeTse:07,OzgurJohariTseLeveque:10}
\begin{equation}
e(\alpha,\beta,\gamma)=\underset{n\rightarrow\infty}\lim \frac{\log
T_n(\alpha,\beta,\gamma)}{\log n}, \label{EQ:exponent}
\end{equation}
which captures the dominant term in the exponent of the throughput
scaling.\footnote{To simplify notations, $T_n(\alpha,\beta,\gamma)$
will be written as $T_n$ if dropping $\alpha$, $\beta$, and $\gamma$
does not cause any confusion.} It is assumed that each node
transmits at a rate $T_n(\alpha,\beta,\gamma)/n$.


\section{Main Results} \label{SEC:result}

This section presents the formal statement of our results, which are
divided into two parts to show the capacity scaling laws: achievable
throughput scalings and information-theoretic upper bounds. We
simply state these results here and derive them in later sections.
The following summarizes our main results. In dense networks, the
optimal scaling exponent is given by $e(\alpha,\beta,\gamma)=1$,
while the optimal scaling exponent $e(\alpha,\beta,\gamma)$ in
extended networks is given by
\begin{eqnarray} \label{EQ:sumEX}
e(\alpha,\beta,\gamma)=\max\left\{1+\gamma-\frac{(1-\beta)\alpha}{2},\min\left\{\beta+\gamma,\frac{\beta+1}{2}\right\},\frac{1}{2},2-\frac{\alpha}{2}\right\},
\end{eqnarray}
where the details are shown in the following two subsections.

\subsection{Achievable Throughput Scaling}

In this subsection, the throughput scaling for both dense and
extended networks under our routing protocols is shown. The
following theorem first presents a lower bound on the total capacity
scaling $T_n$ in dense networks.

\begin{theorem}
In a dense network,
\begin{equation}
T_n=\Omega(n^{1-\epsilon}) \label{EQ:ratedense}
\end{equation}
is achievable with high probability (whp) for an arbitrarily small
$\epsilon>0$.
\end{theorem}

Equation (\ref{EQ:ratedense}) is achievable by simply using the HC
strategy~\cite{OzgurLevequeTse:07}.\footnote{Note that the HC always
outperforms the proposed BS-based routing protocols in terms of
throughput performance under dense networks, even though the details
are not shown in this paper.} Although the HC provides an almost
linear throughput scaling in dense networks, it may degrade
throughput scalings in extended (or power-limited) networks. An
achievable throughput under extended networks is specifically given
as follows.

\begin{theorem} \label{THM:exlower}
In an extended network,
\begin{equation}
T_n=\Omega\left(\max\left\{ml\left(\frac{m}{n}\right)^{\alpha/2-1},m\min\left\{l,\left(\frac{n}{m}\right)^{1/2-\epsilon}\right\},n^{1/2-\epsilon},n^{2-\alpha/2-\epsilon}\right\}\right)
\label{EQ:exrate}
\end{equation}
is achievable whp for an arbitrarily small $\epsilon>0$.
\end{theorem}

The first to fourth terms in~(\ref{EQ:exrate}) correspond to the
achievable rate scalings of the infrastructure-supported single-hop
(ISH), infrastructure-supported multi-hop (IMH), MH, and HC
protocols, respectively, where the two BS-based schemes will be
described in detail later (in Section~\ref{SEC:routing}). As a
result, the best strategy among the four schemes ISH, IMH, MH, and
HC depends on the path-loss exponent $\alpha$, and the parameters
$m$ and $l$ under extended networks. Let us give an intuition for
the achievability result above. For the first term in
(\ref{EQ:exrate}), $ml$ represents the total number of
simultaneously active sources in the ISH protocol while
$(m/n)^{\alpha/2-1}$ comes from a performance degradation due to
power limitation. The second term in (\ref{EQ:exrate}) represents
the total number of sources that can send their own packets
simultaneously using the IMH protocol. From the achievable rates of
each scheme, the interesting result below is obtained under each
network condition.

\begin{remark} \label{REM:extended}
The best achievable one among the four schemes and its scaling
exponent $e(\alpha,\beta,\gamma)$ in (\ref{EQ:exponent}) are shown
in TABLE \ref{T:table1} according to the two-dimensional operating
regimes on the achievable throughput scaling with respect to the
scaling parameters $\beta$ and $\gamma$ (see Fig.
\ref{FIG:scaling_extended}). This result is analyzed in Appendix
\ref{PF:extended}. Operating regimes A--D are shown in Fig.
\ref{FIG:scaling_extended}. It is important to verify the best
protocol in each regime. In Regime A, where $\beta$ and $\gamma$ are
small, the infrastructure is not helpful. In other regimes, we
observe BS-based protocols are dominant in some cases depending on
the path-loss exponent $\alpha$. For example, Regime D has the
following characteristics: the HC protocol has the highest
throughput when the path-loss attenuation is small, but as the
path-loss exponent $\alpha$ increases, the best scheme becomes the
ISH protocol. This is because the penalty for long-range
multiple-input multiple-output (MIMO) transmissions of the HC
increases. Finally, the IMH protocol becomes dominant when $\alpha$
is large since the ISH protocol has a power limitation at the high
path-loss attenuation regime.
\end{remark}

\subsection{Cut-Set Upper Bound}

We now turn our attention to presenting the cut-set upper bound of
the total throughput $T_n$. The upper
bound~\cite{OzgurLevequeTse:07} for pure ad hoc networks of unit
area is extended to our dense network model.

\begin{theorem} \label{THM:denseUpp}
The total throughput $T_n$ is upper-bounded by $n\log n$ whp in a
dense network with infrastructure.
\end{theorem}

Note that the same upper bound as that of~\cite{OzgurLevequeTse:07}
assuming no BSs is found in dense networks. This upper bound means
that $n$ S--D pairs can be active with genie-aided interference
removal between simultaneously transmitting nodes, while providing a
power gain of $\log n$. In addition, it is examined how the upper
bound is close to the achievable throughput scaling.

\begin{remark}
Based on the above result, it is easy to see that the achievable
rate in (\ref{EQ:ratedense}) and the upper bound are of the same
order up to a factor $\log n$ in dense networks with the help of
BSs, and thus the exponent of the capacity scaling is given by
$e(\alpha,\beta,\gamma)=1$. The HC is therefore order-optimal and we
may conclude that infrastructure cannot improve the throughput
scaling in dense networks.
\end{remark}

In extended networks, an upper bound is established based on the
characteristics at power-limited regimes shown
in~\cite{OzgurLevequeTse:07,OzgurJohariTseLeveque:10}, and is
presented in the following theorem.

\begin{theorem} \label{THM:EXupper}
In an extended network, the total throughput $T_n$ is upper-bounded
by
\begin{equation}
T_n=O\left(n^{\epsilon}\max\left\{ml\left(\frac{m}{n}\right)^{\alpha/2-1},m\min\left\{l,\sqrt{\frac{n}{m}}\right\},\sqrt{n},n^{2-\alpha/2}\right\}\right)
\label{EQ:excutset}
\end{equation}
whp for an arbitrarily small $\epsilon>0$.
\end{theorem}

The relationship between the achievable throughput and the cut-set
upper bound is now examined as follows.

\begin{remark}
The upper bound matches the achievable throughput scaling within
$n^{\epsilon}$ in extended networks with infrastructure, and thus
the scaling exponent in (\ref{EQ:sumEX}) holds. In other words, it
is shown that choosing the best of the four schemes ISH, IMH, MH and
HC is order-optimal for all the operating regimes shown in Fig.
\ref{FIG:scaling_extended} (see TABLE \ref{T:table1}).
\end{remark}


\section{Routing Protocols} \label{SEC:routing}

This section explains the two BS-based protocols. Two conventional
schemes~\cite{GuptaKumar:00,OzgurLevequeTse:07} with no
infrastructure support are also introduced for comparison. Each
routing protocol operates in different time slots to avoid huge
mutual interference. We focus on the description for extended
networks since using the HC scheme~\cite{OzgurLevequeTse:07} is
enough to provide a near-optimal throughput in dense networks.

\subsection{Protocols With Infrastructure Support} \label{SEC:withinfra}

We generalize the conventional BS-based transmission scheme
in~\cite{KulkarniViswanath:03,KozatTassiulas:03,ZemlianovVeciana:05,LiuLiuTowsley:03,LiuThiranTowsley:07}:
a source node transmits its packet to the closest BS, the BS having
the packet transmits it to the BS that is nearest to the destination
of the source via wired BS-to-BS links, and the destination finally
receives its data from the nearest BS. Since there exist both access
(to BSs) and exit (from BSs) routings, different time slots are
used, e.g., even and odd time slots, respectively. We start from the
following lemma.

\begin{lemma} \label{LEM:nodenum}
Suppose $m=n^{\beta}$ where $\beta\in[0,1)$. Then, the number of
nodes inside each cell is between
$((1-\delta_0)n^{1-\beta},(1+\delta_0)n^{1-\beta})$, i.e.,
$\Theta(n/m)$, with probability larger than
\begin{equation} \label{EQ:numprob}
1-n^{\beta}e^{-\Delta(\delta_0)n^{1-\beta}},
\end{equation}
where $\Delta(\delta_0)=(1+\delta_0)\ln(1+\delta_0)-\delta_0$ for
$0<\delta_0<1$ independent of $n$.
\end{lemma}

The proof of this lemma is given by slightly modifying the proof of
Lemma 4.1 in~\cite{OzgurLevequeTse:07}. Note that (\ref{EQ:numprob})
tends to one as $n$ goes to infinity.

\subsubsection{Infrastructure-supported single-hop (ISH) protocol}

In contrast with previous works, the spatial dimensions enabled by
having multiple antennas at each BS are exploited here, and thus
multiple transmission pairs can be supported using a single BS.
Under extended networks, the ISH transmission protocol shown in Fig.
\ref{FIG:ISH} is now proposed as follows:
\begin{itemize}
\item For the access routing, all source nodes in each cell, given by $n/m$ nodes whp from Lemma \ref{LEM:nodenum}, transmit their independent packets simultaneously via single-hop multiple-access to the BS in the same
cell.
\item Each BS receives and jointly decodes packets from source nodes in the same
cell. Signals received from the other cells are treated as noise.
\item The BS that completes decoding its packets transmits them to the BS closest to the corresponding destination by wired BS-to-BS links.
\item For the exit routing, each BS transmits all packets received from other BSs,
i.e., $n/m$ packets, via single-hop broadcast to the destinations in
the cell.
\end{itemize}

Since the network is power-limited, the proposed ISH scheme is used
with the full power, i.e., the transmit powers at each node and BS
are $P$ and $\frac{nP}{m}$, respectively.

For the ISH protocol, more DoF gain is provided compared to
transmissions via MH if $m$ and $l$ are large enough. The power gain
can also be obtained compared to that of the HC scheme in certain
cases.

\subsubsection{Infrastructure-supported multi-hop (IMH) protocol}

The fact that the extended network is power-limited motivates the
introduction of an IMH transmission protocol in which multiple
source nodes in a cell transmit their packets to BS in the cell via
MH, thereby having much higher received power, i.e., more power
gain, than that of the direct one-hop transmission in extended
networks. That is, better long-distance SNR is provided with the IMH
protocol. Similarly, each BS delivers its packets to the
corresponding destinations by IMH transmissions. Under extended
networks, the IMH transmission protocol in Fig. \ref{FIG:IMH} is
proposed as follows:
\begin{itemize}
\item Divide each cell into smaller square cells of area $2\log n$ each, where these smaller cells are called routing
cells (which include at least one node
whp~\cite{GuptaKumar:00,ElGamalMammenPrabhakarShah:06}).
\item For the access routing, $\min\{l,\sqrt{n/m}\}$ source nodes in each cell transmit their independent packets using MH routing to
the corresponding BS in the cell as shown in Fig.~\ref{FIG:access}.
It is assumed that each antenna placed only on the BS boundary
receives its packet from one of the nodes in the nearest neighbor
routing cell. It is easy to see that $\min\{l,\sqrt{n/m}\}$ MH paths
can be supported due to our antenna placement within BSs. Exit
routing is similar, where each antenna on the BS boundary uses power
$P$ that satisfies the power constraint.
\item The BS-to-BS transmissions are the same as the ISH case.
\item Each routing cell operates based on $9$-time division multiple access (TDMA) to avoid causing huge interference to its neighbor cells.
\end{itemize}

Note that the transmit power $\min\{l,\sqrt{n/m}\}$ at each BS, but
not full power, is enough to perform the IMH protocol in the
downlink.

For the IMH protocol, more DoF gain is possible compared to the MH
scheme for large $m$ and $l$. In addition, more power gain can also
be obtained compared to the HC and ISH schemes in certain cases.

\subsection{Protocols Without Infrastructure Support}
\label{SEC:withoutinfra}

The upper bound in Theorem~\ref{THM:denseUpp} is only determined by
the number $n$ of wireless nodes in dense networks. The upper bound
in Theorem~\ref{THM:EXupper} also indicates that either the number
$m$ of BSs or the number $l$ of antennas per BS should be greater
than a certain level in order to obtain improved throughput scalings
in extended networks. This is because otherwise less DoF gain may be
provided compared to that of the conventional schemes without help
of BSs. Thus, transmissions only using wireless nodes may be enough
to achieve the capacity scalings in dense networks or in extended
networks with small $m$ and $l$. In this case, the MH and HC
protocols, which were proposed in~\cite{GuptaKumar:00}
and~\cite{OzgurLevequeTse:07}, respectively, are performed in our
network with infrastructure.


\section{Achievable Throughput in Extended Networks} \label{SEC:scaling}

In this section, the achievable throughput scaling in
Theorem~\ref{THM:exlower} is rigorously analyzed in extended
networks. It is demonstrated that the throughput scaling can be
improved under some conditions by applying two BS-based
transmissions in extended networks.

The transmission rate of the ISH protocol in extended networks will
be shown first.

\begin{lemma} \label{LEM:accessISH}
Suppose that an extended network uses the ISH protocol. Then, the
rate of
\begin{equation}
\Omega\left(l\left(\frac{m}{n}\right)^{\alpha/2-1}\right) \nonumber
\end{equation}
is achievable at each cell for both access and exit routings.
\end{lemma}

\begin{proof}
In order to prove the result, we need to quantify the total amount
of interference when the ISH scheme is used. We first introduce the
following lemma and refer to Appendix~\ref{PF:intISH} for the
detailed proof.

\begin{lemma} \label{LEM:intISH}
In an extended network with the ISH protocol, the total interference
power $P_I^u$ in the uplink from nodes in other cells to each BS
antenna is upper-bounded by $\Theta((m/n)^{\alpha/2-1})$ whp. Each
node also has interference power $P_I^d=\Theta((m/n)^{\alpha/2-1})$
whp in the downlink from BSs in other cells.
\end{lemma}

Note that when $\alpha>2$ the term $(m/n)^{\alpha/2-1}$ tends to
zero as $n\rightarrow\infty$. Now, the transmission rate for the
access routing is derived as in the following. The signal model from
nodes in each cell to the BS with multiple antennas corresponds to
the single-input multiple-output (SIMO) multiple-access channel
(MAC). Since the maximum Euclidean distance among links of the above
SIMO MAC scales as $\Theta(\sqrt{n/m})$, it is upper-bounded by as
$\delta_1\sqrt{n/m}$, where $\delta_1>0$ is a certain constant. Let
$N_I$ denote the sum of total interference power $P_I^u$ received
from the other cells and noise variance $N_0$. Then, the worst case
noise of this channel has an uncorrelated Gaussian distribution with
zero-mean and variance
$N_I$~\cite{Medard:00,DiggaviCover:01,HassibiHochwald:03}, which
lower-bounds the transmission rate. By assuming full CSI at the
receiver (BS $s$) and performing a minimum mean-square error (MMSE)
estimation~\cite{VaranasiGuess:97,ViswanathTse:03,TseViswanath:05}
with successive interference cancellation (SIC) at BS $s$, the
sum-rate of the SIMO MAC is given
by~\cite{ViswanathTse:03,TseViswanath:05}
\begin{eqnarray} \label{EQ:SIMOMAC}
I({\bf x}_s;{\bf y}_s,{\bf H}_s)\!\!\!\!\!\!\!&&\ge
E\left[\log\det\left({\bf I}_l+\frac{P}{N_I}{\bf H}_s{\bf
H}_s^{\dagger}\right)\right] \nonumber \\ &&\ge
 E\left[\log\det\left({\bf
I}_l+\frac{P}{\delta_1^{\alpha}(n/m)^{\alpha/2}N_I}{\bf G}{\bf
G}^{\dagger}\right)\right],
\end{eqnarray}
where ${\bf x}_s$ denotes the $\frac{n}{m}\times1$ transmit signal
vector, whose elements are nodes in the cell covered by BS $s$,
${\bf y}_s$ is the $l\times1$ received signal vector at BS $s$, and
${\bf H}_s=[{\bf h}_{s1}^u \ {\bf h}_{s2}^u \ \cdots \ {\bf
h}_{s(n/m)}^u]$ (${\bf h}_{si}^u$ for $i=1,\cdots,n/m$ is given in
(\ref{EQ:signalU})). ${\bf G}$ is the normalized matrix, whose
element $g_{ti}$ is given by $e^{j\theta_{si,t}^u}$ and represents
the phase between node $i$ and the $t$-th antenna of BS $s$. Note
that rotating the decoding order among $n/m$ nodes in the cell leads
to the same rate of each node. Then, the above sum-rate is rewritten
as
\begin{eqnarray}
I({\bf x}_s;{\bf y}_s,{\bf H}_s)\!\!\!\!\!\!\!\!&&\ge
E\left[\sum_{i=1}^{l}\log\left(1+\frac{P}{\delta_1^{\alpha}(n/m)^{\alpha/2-1}N_I}\lambda_i\right)\right]
\nonumber\\ &&\ge l
E\left[\log\left(1+\frac{P}{\delta_1^{\alpha}(n/m)^{\alpha/2-1}N_I}\lambda_1\right)\right]
\nonumber
\\ &&\ge
l
\log\left(1+\frac{P}{\delta_1^{\alpha}(n/m)^{\alpha/2-1}N_I}\bar{\lambda}\right)\Pr\left(\lambda_1>\bar{\lambda}\right),
\label{EQ:MIlambda}
\end{eqnarray}
where $\{\lambda_1,\cdots,\lambda_l\}$ are the unordered eigenvalues
of $\frac{m}{n}{\bf G}{\bf G}^{\dagger}$~\cite{Telatar:99} and
$\bar{\lambda}$ is any nonnegative constant. Due to the fact that
$\log(1+x)=(\log e)x+O(x^2)$ for small $x>0$, (\ref{EQ:MIlambda}) is
given by
\begin{equation} \label{EQ:MIlambda2}
I({\bf x}_s;{\bf y}_s,{\bf H}_s)\ge c_0 l
\left(\frac{m}{n}\right)^{\alpha/2-1}\Pr\left(\lambda_1>\bar{\lambda}\right)
\end{equation}
for some constant $c_0>0$ independent of $n$, since $N_I$ has a
constant scaling from Lemma~\ref{LEM:intISH}. By the Paley-Zygmund
inequality~\cite{Kahane:85}, it is possible to lower-bound the
sum-rate in the left-hand side (LHS) of (\ref{EQ:MIlambda2}) by
following the same line as Appendix I in~\cite{OzgurLevequeTse:07},
thus yielding
\begin{equation}
I({\bf x}_s;{\bf y}_s,{\bf H}_s) \ge
c_1l\left(\frac{m}{n}\right)^{\alpha/2-1}, \nonumber
\end{equation}
where $c_1>0$ is some constant independent of $n$.

For the exit routing, the signal model from the BS with multiple
antennas in one cell to nodes in the cell corresponds to the
multiple-input single-output (MISO) broadcast channel (BC). From
Lemma \ref{LEM:intISH}, it is seen that the total interference power
received from the other BSs is bounded. Hence, it is possible to
derive the transmission rate for the exit routing by exploiting an
uplink-downlink
duality~\cite{ViswanathTse:03,TseViswanath:05,ViswanathJindalGoldsmith:03,Yu:06}.
In this case, the transmitters in the downlink are designed by an
MMSE transmit precoding with dirty paper
coding~\cite{Costa:83,CaireShamai:03,WeingartenSteinbergShamai:06}
at each BS, and the rate of the MISO BC is then equal to that of the
dual SIMO MAC with a sum power constraint. More precisely, with full
CSI at the transmitter (BS) and the total transmit power
$\frac{nP}{m}$ in the downlink, the sum-rate of the MISO BC is
lower-bounded by~\cite{ViswanathTse:03}
\begin{eqnarray} \label{EQ:MISOBC}
&& \underset{{\bf Q}_x\ge0}\max E\left[\log\det\left({\bf
I}_l+\frac{1}{N_I'}{{\bf
H}_s'}^{\dagger}{\bf Q}_x{{\bf H}_s'}\right)\right] \nonumber\\
&&\ge  E\left[\log\det\left({\bf I}_l+\frac{P}{N_I'}{{\bf
H}_s'}^{\dagger}{{\bf H}_s'}\right)\right],
\end{eqnarray}
where ${{\bf H}_s'}=[{\bf h}_{1s}^{d~\!T} \ {\bf h}_{2s}^{d~\!T} \
\cdots \ {\bf h}_{(n/m)s}^{d~T}]^T$, $N_I'$ denotes the sum of total
interference power $P_I^d$ from BSs in the other cells and noise
variance $N_0$, and ${\bf Q}_x$ is the
$\frac{n}{m}\times\frac{n}{m}$ positive semi-definite input
covariance matrix which is diagonal and satisfies $\tr({\bf Q}_x)\le
\frac{nP}{m}$. Here, the inequality holds since the rate is reduced
by simply applying the same average power of each user. Due to the
fact that (\ref{EQ:MISOBC}) is equivalent to the right-hand side
(RHS) of (\ref{EQ:SIMOMAC}) (with a change of variables),
$\Omega\left(l(m/n)^{\alpha/2-1}\right)$ is achievable in the
downlink of each cell by following the same approach as that for the
access routing, which completes the proof of
Lemma~\ref{LEM:accessISH}.
\end{proof}

Note that $l$ corresponds to the DoF at each cell provided by the
ISH protocol while $(m/n)^{\alpha/2-1}$ represents the throughput
degradation due to power loss.

The transmission rate for the access and exit routings of IMH
protocol will now be analyzed in extended networks. The number of
source nodes that can be active simultaneously is examined under the
IMH protocol, while maintaining a constant throughput $\Theta(1)$
per S--D pair.

\begin{lemma} \label{LEM:IMH}
When an extended network uses the IMH protocol, the rate of
\begin{equation}
\Omega\left(\min\left\{l,\left(\frac{n}{m}\right)^{1/2-\epsilon}\right\}\right)
\label{EQ:rateIMH}
\end{equation}
is achievable at each cell for both access and exit routings, where
$\epsilon>0$ is an arbitrarily small constant.
\end{lemma}

\begin{proof}
This result is obtained by modifying the analysis
in~\cite{GuptaKumar:00,ElGamalMammenPrabhakarShah:06,ShinChungLee:09}
on scaling laws under our BS-based network. We mainly focus on the
aspects that are different from the conventional schemes. From the
9-TDMA operation, the signal-to-interference-and-noise ratio (SINR)
seen by any receiver is given by $\Omega(1)$ with a transmit power
$P$. It can be interpreted that when the worst case
noise~\cite{Medard:00,DiggaviCover:01,HassibiHochwald:03} is assumed
as in the ISH protocol, the achievable throughput per S--D pair is
lower-bounded by $\log(1+\text{SINR})$, thus providing a constant
scaling. First consider the case $l=o(\sqrt{n/m})$ where the number
$l$ of antennas scales slower than the number $n/m$ of nodes in a
cell. Then, it is possible to activate up to $l$ source nodes at
each cell because there exist $l$ routes for the last hop to each BS
antenna in the uplink. On the other hand, when
$l=\Omega(\sqrt{n/m})$, the maximum number of simultaneously
transmitting sources per BS is equal to the number of routing cells
on the BS boundary, which scales with $(n/m)^{1/2-\epsilon}$ for an
arbitrarily small $\epsilon>0$. In the downlink of each cell, the
same number of S--D pairs as that in the uplink is active
simultaneously. Therefore, the transmission rate per each BS is
finally given by (\ref{EQ:rateIMH}), which completes the proof of
Lemma~\ref{LEM:IMH}.
\end{proof}

By using Lemmas~\ref{LEM:accessISH}~and~\ref{LEM:IMH}, we are ready
to show the achievable throughput scaling in extended networks. The
achievable throughputs of the ISH and IMH protocols are given by
\begin{equation}
T_n=\Omega\left(ml\left(\frac{m}{n}\right)^{\alpha/2-1}\right)
\label{EQ:ISHex}
\end{equation}
and
\begin{equation}
T_n=\Omega\left(m\min\left\{l,\left(\frac{n}{m}\right)^{1/2-\epsilon}\right\}\right),
\label{EQ:IMHex}
\end{equation}
respectively, since there are $m$ cells in the network. Throughput
scalings of two conventional protocols that do not utilize the BSs
are also considered. From the results
of~\cite{GuptaKumar:00,OzgurLevequeTse:07},
\begin{equation} \label{EQ:exMH}
T_n=\Omega\left(n^{1/2-\epsilon}\right)
\end{equation}
and
\begin{equation} \label{EQ:exHC}
T_n=\Omega(n^{2-\alpha/2-\epsilon})
\end{equation}
are yielded for the MH and HC schemes, respectively. Hence, the
throughput scaling in extended networks is simply lower-bounded by
the maximum of (\ref{EQ:ISHex})--(\ref{EQ:exHC}), which completes
the proof of Theorem~\ref{THM:exlower}.

In addition, we would like to examine the required rate of each
BS-to-BS transmission.

\begin{remark} \label{REM:CBS}
To see how much data traffic flows on each BS-to-BS link, we first
show the following lemma.

\begin{lemma} \label{LEM:numBSlink}
Let $X_{ki}$ denote the number of destinations in the $k$-th cell
whose source nodes are in the $i$-th cell, where $i, k\in
\{1,\cdots,m\}$. Then, for all $i,k\in\{1,\cdots,m\}$, the following
equation holds whp:
\begin{eqnarray} \label{EQ:Xki}
X_{ki}=\left\{\begin{array}{lll} O\left(\frac{n}{m^2}\right)
&\textrm{if
$n=\omega(m^2)$} \\
O(\log n) &\textrm{if $n=O(m^2)$}
\end{array}\right..
\end{eqnarray}
\end{lemma}

The proof of this lemma is presented in Appendix~\ref{PF:numBSlink}.
Let $C_{BS}$ denote the rate of each BS-to-BS link. Then, since each
S--D pair transmits at a rate $T_n/n$ and the number of packets
carried simultaneously through each link is bound by (\ref{EQ:Xki})
from Lemma \ref{LEM:numBSlink}, the required rate $C_{BS}$ is given
by
\begin{eqnarray}
C_{BS}=\left\{\begin{array}{lll} \Omega\left(\frac{T_n}{m^2}\right)
&\textrm{if
$n=\omega(m^2)$} \\
\Omega\left(\frac{T_n\log n}{n}\right) &\textrm{if $n=O(m^2)$}
\end{array}\right.. \label{EQ:CBS} \nonumber
\end{eqnarray}
\end{remark}


\section{Cut-Set Upper Bound} \label{SEC:upper}

To see how closely the proposed schemes approach the fundamental
limits in a network with infrastructure, new BS-based cut-set outer
bounds on the throughput scaling are analyzed based on the
information-theoretic approach~\cite{CoverThomas:91}.

\subsection{Dense Networks}

Before showing the main proof of Theorem~\ref{THM:denseUpp}, we
start from the following lemma.

\begin{lemma} \label{LEM:distance}
In our two-dimensional dense network where $n$ nodes are uniformly
distributed and there are $m$ BSs with $l$ regularly spaced
antennas, the minimum distance between any two nodes or between a
node and an antenna on the BS boundary is larger than
$1/n^{1+\epsilon_1}$ whp for an arbitrarily small $\epsilon_1>0$.
\end{lemma}

The proof of this lemma is presented in Appendix~\ref{PF:distance}.
Now we present the cut-set upper bound of the total throughput $T_n$
in dense networks. The proof steps are similar to those
of~\cite{GastparVetterli:05,OzgurLevequeTse:07}. The throughput per
S--D pair is simply upper-bounded by the capacity of the SIMO
channel between a source node and the rest of nodes including BSs.
Hence, the total throughput for $n$ S--D pairs is bounded by
\begin{eqnarray}
T_n\!\!\!\!\!\!\!&&\le\sum_{i=1}^n\log\left(1+\frac{P}{N_0}\left(\sum_{k=1
\atop k\neq i}^{n}|h_{ki}|^2+\sum_{s=1}^{m}\|{\bf
h}_{si}^u\|^2\right)\right) \nonumber
\\ \!\!\!\!\!\!\!&&\le
n\log\left(1+\frac{P}{N_0}n^{(1+\epsilon_1)\alpha}(n-1+ml)\right)
\nonumber
\\ \!\!\!\!\!\!\!&&=c_2 n\log n, \nonumber
\end{eqnarray}
where $\|\cdot\|$ denotes $L_2$-norm of a vector and $c_2>0$ is some
constant independent of $n$. The second inequality holds due to
Lemma \ref{LEM:distance}. This completes the proof of
Theorem~\ref{THM:denseUpp}.

\subsection{Extended Networks}

Consider the cut $L$ in Fig.~\ref{FIG:cutset} dividing the network
area into two halves in an extended random network. Let $S_{L}$ and
$D_{L}$ denote the sets of sources and destinations, respectively,
for the cut in the network. More precisely, under $L$, (wireless)
source nodes $S_{L}$ are on the left half of the network, while all
nodes on the right half and all BS antennas are destinations
$D_{L}$.\footnote{The other cut $\tilde{L}$ can also be considered
in the network. In this case, sources $S_{\tilde{L}}$ represent
antennas at each BS as well as ad hoc nodes on the left half. The
(wireless) destination nodes $D_{\tilde{L}}$ are on the right half.
Since the cut $L$ provides a tight upper bound compared to the
achievable rate, the analysis for the cut $\tilde{L}$ is not shown
in this paper.} In this case, we get the $n\times (n+ml)$ MIMO
channel between the two sets of nodes and BSs separated by the cut.

In extended networks, it is necessary to narrow down the class of
S--D pairs according to their Euclidean distance to obtain a tight
upper bound. In this subsection, the upper bound based on the power
transfer arguments
in~\cite{OzgurLevequeTse:07,OzgurJohariTseLeveque:10} is shown,
where an upper bound is proportional to the total received signal
power from source nodes. The present problem is not equivalent to
the conventional extended setup since a network with infrastructure
support is taken into account. A new upper bound based on hybrid
approaches that consider either the sum of the capacities of the
multiple-input single-output (MISO) channel between transmitters and
each receiver or the amount of power transferred across the network
according to operating regimes, is thus derived. We start from the
following lemma.

\begin{lemma} \label{LEM:randomlog}
Assume a two-dimensional extended network. When the network area
with the exclusion of BS area is divided into $n$ squares of unit
area, there are less than $\log n$ nodes in each square whp.
\end{lemma}

This result can be obtained by applying our BS-based network and
slightly modifying the proof of Lemma 1
in~\cite{FranceschettiDouseTseThiran:07}. For the cut $L$, the total
throughput $T_n$ for sources on the left half is bounded by the
capacity of the MIMO channel between $S_{L}$ and $D_{L}$, and thus
\begin{eqnarray}
T_n \!\!\!\!\!\!\!\!&&\le\underset{{\bf Q}_{L}\ge 0} \max
E\left[\log\det\left({\bf I}_{n+ml}+{\bf H}_{L}{\bf Q}_{L}{\bf
H}_{L}^{\dagger}\right)\right] \nonumber\\ &&=\underset{{\bf
Q}_{L}\ge 0} \max E\left[\log\det\left({\bf I}_{\Theta(n)}+{\bf
H}_{L}{\bf Q}_{L}{\bf H}_{L}^{\dagger}\right)\right], \nonumber
\end{eqnarray}
where the equality holds since $n=\Omega(ml)$.\footnote{Here and in
the sequel, the noise variance $N_0$ is assumed to be 1 to simplify
the notation.} ${\bf H}_{L}$ consists of ${\bf h}_{si}^u$ in
(\ref{EQ:signalU}) for $i\in S_{L}$, $s\in B$, and $h_{ki}$ in
(\ref{EQ:signalW}) for $i\in S_{L}$, $k\in D_{r}$. Here, $B$ and
$D_{r}$ represent the set of BSs in the network and the set of
(wireless) nodes on the right half, respectively. ${\bf Q}_{L}$ is
the positive semi-definite input covariance matrix whose $k$-th
diagonal element satisfies $[{\bf Q}_{L}]_{kk}\le P$ for $k\in
S_{L}$. The set $D_{L}$ ($=B\cup D_r$) is partitioned into three
groups according to their location, as shown in Fig. \ref{FIG:Dcut}.
By generalized Hadamard's inequality~\cite{ConstantinescuScharf:98}
as in~\cite{JovicicViswanathKulkarni:04,OzgurLevequeTse:07},
\begin{eqnarray}
T_n\!\!\!\!\!\!&&\le \underset{{\bf Q}_{L}\ge 0} \max
E\left[\log\det\left({\bf I}_{\sqrt{n}}+{\bf H}_{L}^{(1)}{\bf
Q}_{L}{{\bf
H}_{L}^{(1)}}^{\dagger}\right)\right] \nonumber\\
&&+\underset{{\bf Q}_{L}\ge 0} \max E\left[\log\det\left({\bf
I}_{O(\sqrt{ml})}+{\bf H}_{L}^{(2)}{\bf Q}_{L}{{\bf
H}_{L}^{(2)}}^{\dagger}\right)\right]\nonumber\\&&+\underset{{\bf
Q}_{L}\ge 0} \max E\left[\log\det\left({\bf I}_{\Theta(n)}+{\bf
H}_{L}^{(3)}{\bf Q}_{L}{{\bf H}_{L}^{(3)}}^{\dagger}\right)\right],
\label{EQ:thcut1}
\end{eqnarray}
where ${\bf H}_{L}^{(t)}$ is the matrix with entries $\left[{\bf
H}_{L}^{(t)}\right]_{ki}$ for $i\in S_{L}$, $k\in D_{L}^{(t)}$, and
$t=1,\cdots,3$. Here, $D_{L}^{(1)}$ and $D_{L}^{(2)}$ denote the
sets of destinations located on the rectangular slab with width 1
immediately to the right of the centerline (cut) and on the ring
with width 1 immediately inside each BS boundary (cut) on the left
half, respectively. $D_{L}^{(3)}$ is given by
$D_{L}\setminus(D_{L}^{(1)}\cup D_{L}^{(2)})$. Note that the sets
($D_{L}^{(1)}$ and $D_{L}^{(2)}$) of destinations located very close
to the cut are considered separately since otherwise their
contribution to the total received power will be excessive,
resulting in a loose bound.

Each term in (\ref{EQ:thcut1}) will be analyzed below in detail.
Before that, to get the total power transfer of the set
$D_{L}^{(3)}$, the same technique as that in Section V
of~\cite{OzgurLevequeTse:07} is used, which is the relaxation of the
individual power constraints to a total weighted power constraint,
where the weight assigned to each source corresponds to the total
received power on the other side of the cut. Specifically, each
column $i$ of the matrix ${\bf H}_{L}^{(3)}$ is normalized by the
square root of the total received power on the other side of the cut
from source $i\in S_{L}$. The total weighted power $P_{L,i}^{(3)}$
by source $i$ is then given by
\begin{equation}
P_{L,i}^{(3)}=Pd_{L,i}^{(3)}, \label{EQ:Ptots1}
\end{equation}
where
\begin{equation}
d_{L,i}^{(3)}=\underset{k\in\bar{D}_r \setminus D_{L}^{(1)}}\sum
r_{ki}^{-\alpha}+\underset{s\in B_l, t\in[1,l]}\sum
r_{si,t}^{u~\!-\alpha}. \label{EQ:dsum1}
\end{equation}
Here, $\bar{D}_r$ is the set of destination nodes including BS
antennas on the right half and $B_l$ represents the set of BSs on
the left half. Then, the third term in (\ref{EQ:thcut1}) is
rewritten as
\begin{equation}
\underset{\tilde{\bf Q}_{L}\ge 0} \max E\left[\log\det\left({\bf
I}_{n}+{\bf F}_{L}^{(3)}\tilde{\bf Q}_{L}{{{\bf
F}}{}_{L}^{(3)}}^{\dagger}\right)\right], \label{EQ:cut1third}
\end{equation}
where ${\bf F}_{L}^{(3)}$ is the matrix with entries $\left[{\bf
F}_{L}^{(3)}\right]_{ki}=\frac{1}{\sqrt{d_{L,i}^{(3)}}}\left[{\bf
H}_{L}^{(3)}\right]_{ki}$, which are obtained from (\ref{EQ:dsum1}),
for $i\in S_{L}$, $k\in D_{L}^{(3)}$. Then, $\tilde{\bf Q}_{L}$ is
the matrix satisfying
\begin{equation}
\left[\tilde{\bf
Q}_{L}\right]_{ki}=\sqrt{d_{L,k}^{(3)}d_{L,i}^{(3)}}\left[{\bf
Q}_{L}\right]_{ki}, \nonumber
\end{equation}
which means $\tr(\tilde{\bf Q}_{L})\le\sum_{i\in S_{L}}
P_{L,i}^{(3)}$ (equal to the sum of the total received power from
each source).

We next examine the behavior of the largest singular value for the
normalized channel matrix ${\bf F}_{L}^{(3)}$. From the fact that
${\bf F}_{L}^{(3)}$ is well-conditioned whp, this shows how much it
essentially affects an upper bound of (\ref{EQ:cut1third}), which
will be analyzed later in Lemma \ref{LEM:uppermaxsum}.

\begin{lemma} \label{LEM:Hsing}
Let ${\bf F}_{L}^{(3)}$ denote the normalized channel matrix whose
element is given by $\left[{\bf
F}_{L}^{(3)}\right]_{ki}=\frac{1}{\sqrt{d_{L,i}^{(3)}}}\left[{\bf
H}_{L}^{(3)}\right]_{ki}$. Then,
\begin{equation}
E\left[\left\|{\bf F}_{L}^{(3)}\right\|_2^2\right]\le c_3(\log n)^3,
\label{EQ:ineqHLp}
\end{equation}
where $\|\cdot\|_2$ denotes the largest singular value of a matrix
and $c_3>0$ is some constant independent of $n$.
\end{lemma}

The proof of this lemma is presented in Appendix \ref{PF:Hsing}.
Using Lemma \ref{LEM:Hsing} yields the following result.

\begin{lemma} \label{LEM:uppermaxsum}
The term shown in (\ref{EQ:cut1third}) is upper-bounded by
\begin{equation}
n^{\epsilon}\underset{i\in S_{L}}\sum P_{L,i}^{(3)}
\label{EQ:Ptransfer}
\end{equation}
whp where $\epsilon>0$ is an arbitrarily small constant and
$P_{L,i}^{(3)}$ is given by (\ref{EQ:Ptots1}).
\end{lemma}

\begin{proof}
Equation (\ref{EQ:cut1third}) is bounded by
\begin{eqnarray}
\underset{\tilde{\bf Q}_{L}\ge 0} \max E\left[\log\det\left({\bf
I}_{n}+{\bf F}_{L}^{(3)}\tilde{\bf Q}_{L}{{{\bf
F}}{}_{L}^{(3)}}^{\dagger}\right)1_{\mathcal{E}_{{\bf
F}_{L}^{(3)}}}\right]\nonumber\\ +\underset{\tilde{\bf Q}_{L}\ge 0}
\max E\left[\log\det\left({\bf I}_{n}+{\bf F}_{L}^{(3)}\tilde{\bf
Q}_{L}{{{\bf F}}{}_{L}^{(3)}}^{\dagger}\right)1_{\mathcal{E}_{{\bf
F}_{L}^{(3)}}^{c}}\right], \label{EQ:upperE}
\end{eqnarray}
where the event $\mathcal{E}_{{\bf F}_{L}^{(3)}}$ is given by
\begin{equation}
\mathcal{E}_{{\bf F}_{L}^{(3)}}=\left\{\left\|{\bf
F}_{L}^{(3)}\right\|_2^2>n^{\epsilon}\right\} \nonumber
\end{equation}
for an arbitrarily small constant $\epsilon>0$. Then, by using the
result of Lemma~\ref{LEM:Hsing} and applying the proof technique
similar to that in Section V of~\cite{OzgurLevequeTse:07}, it is
possible to prove that the first term in (\ref{EQ:upperE}) decays
polynomially to zero as $n$ tends to infinity, and for the second
term in (\ref{EQ:upperE}), it follows that
\begin{eqnarray}
\underset{\tilde{\bf Q}_{L}\ge 0} \max E\left[\log\det\left({\bf
I}_{n}+{\bf F}_{L}^{(3)}\tilde{\bf Q}_{L}{{{\bf
F}}{}_{L}^{(3)}}^{\dagger}\right)1_{\mathcal{E}_{{\bf
F}_{L}^{(3)}}^{c}}\right] \le n^{\epsilon}\underset{i\in S_{L}}\sum
P_{L,i}^{(3)}, \nonumber
\end{eqnarray}
which completes the proof.
\end{proof}

Note that (\ref{EQ:Ptransfer}) represents the power transfer from
the set $S_{L}$ of sources to the set $D_{L}^{(3)}$ of the
corresponding destinations for a given cut $L$. For notational
convenience, let $d_{L,i}^{(4)}$ and $d_{L,i}^{(5)}$ denote the
first and second terms in (\ref{EQ:dsum1}), respectively. Then,
$Pd_{L,i}^{(4)}$ and $Pd_{L,i}^{(5)}$ correspond to the total
received power from source $i$ to the destination sets
$\bar{D}_r\setminus D_{L}^{(1)}$ and $D_{L}\setminus(D_{L}^{(2)}
\cup \bar{D}_r)$, respectively. The computation of the total
received power of the set $D_{L}^{(3)}$ will now be computed as
follows:
\begin{eqnarray}
\underset{i\in S_{L}}\sum P_{L,i}^{(3)}=\underset{i\in S_{L}}\sum
Pd_{L,i}^{(4)}+\underset{i\in S_{L}}\sum Pd_{L,i}^{(5)},
\label{EQ:sum345}
\end{eqnarray}
which is eventually used to compute (\ref{EQ:Ptransfer}).

First, to get an upper bound on $\sum_{i\in S_{L}}Pd_{L,i}^{(4)}$ in
(\ref{EQ:sum345}), the network area is divided into $n$ squares of
unit area. By Lemma \ref{LEM:randomlog}, since there are less than
$\log n$ nodes inside each square whp, the power transfer can be
upper-bounded by that under a regular network with at most $\log n$
nodes at each square (see~\cite{OzgurLevequeTse:07} for the detailed
description). Such a modification yields the following upper
bound~\cite{OzgurLevequeTse:07} for $\sum_{i\in
S_{L}}Pd_{L,i}^{(4)}$:
\begin{eqnarray}
\underset{i\in S_{L}}\sum Pd_{L,i}^{(4)}\le
\left\{\begin{array}{lll} c_4n^{2-\alpha/2}(\log n)^2 &\textrm{if
$2<\alpha<3$} \\
c_4\sqrt{n}(\log n)^3 &\textrm{if $\alpha=3$}
\\
c_4\sqrt{n}(\log n)^2 &\textrm{if $\alpha>3$}
\end{array}\right. \label{EQ:PdLi5}
\end{eqnarray}
whp for some constant $c_4>0$ independent of $n$. Next, the second
term $\sum_{i\in S_{L}}Pd_{L,i}^{(5)}$ in (\ref{EQ:sum345}) can be
derived as in the following lemma.

\begin{lemma} \label{LEM:EXdistance}
The term $\sum_{i\in S_{L}}Pd_{L,i}^{(5)}$ is given by
\begin{equation}
\sum_{i\in S_{L}}Pd_{L,i}^{(5)}=\left\{\begin{array}{llll} 0
&\textrm{if
$l=o(\sqrt{n/m})$} \\
O\left(nl\left(\frac{m}{n}\right)^{\alpha/2}\log n\right)
&\textrm{if $l=\Omega(\sqrt{n/m})$ and $2<\alpha<3$} \\
O\left(ml\sqrt{\frac{m}{n}}(\log n)^2\right) &\textrm{if
$l=\Omega(\sqrt{n/m})$ and $\alpha=3$} \\
O\left(\frac{n}{\sqrt{l}}\left(\frac{ml}{n}\right)^{\alpha/2}\log
n\right) &\textrm{if $l=\Omega(\sqrt{n/m})$ and
$\alpha>3$}\end{array}\right. \label{EQ:PL1sum}
\end{equation}
whp.
\end{lemma}

The proof of this lemma is presented in Appendix
\ref{PF:EXdistance}.

It is now possible to derive the cut-set upper bound in
Theorem~\ref{THM:EXupper} by using Lemmas~\ref{LEM:uppermaxsum} and
\ref{LEM:EXdistance}. For notational convenience, let $T_{n}^{(i)}$
denote the $i$-th term in the RHS of (\ref{EQ:thcut1}) for
$i\in\{1,2,3\}$. By generalized Hadamard's
inequality~\cite{ConstantinescuScharf:98} as
in~\cite{JovicicViswanathKulkarni:04,OzgurLevequeTse:07}, the first
term $T_{n}^{(1)}$ in (\ref{EQ:thcut1}) can be easily bounded by
\begin{eqnarray}
T_{n}^{(1)}\!\!\!\!\!\!\!&&\le \underset{k\in
D_{L}^{(1)}}\sum\log\left(1+\frac{P}{N_0}\underset{i\in S_{L}}\sum
|h_{ki}|^2\right) \nonumber \\ &&\le \bar{c}_1\sqrt{n}(\log n)^2,
\label{EQ:TnL11}
\end{eqnarray}
where $\bar{c}_1>0$ is some constant independent of $n$. Note that
this upper bound does not depend on $\beta$ and $\gamma$. The second
inequality holds since the minimum distance between any source and
destination is larger than $1/n^{1/2+\epsilon_1}$ whp for an
arbitrarily small $\epsilon_1>0$, which is obtained by the
derivation similar to that of Lemma \ref{LEM:distance}, and there
exist no more than $\sqrt{n}\log n$ nodes in $D_{L}^{(1)}$ whp by
Lemma \ref{LEM:randomlog}. The upper bound for $T_n^{(2)}$ is now
derived. Since some nodes in $D_{L}^{(2)}$ are located very close to
the cut and the information transfer to $D_{L}^{(2)}$ is limited in
DoF, the second term $T_{n}^{(2)}$ of (\ref{EQ:thcut1}) is
upper-bounded by the sum of the capacities of the MISO channels.
More precisely, by generalized Hadamard's inequality,
\begin{eqnarray} \label{EQ:TnL12}
T_{n}^{(2)}\!\!\!\!\!\!\!\!&&\le \left\{\begin{array}{lll} \bar{c}_2
ml\log n &\textrm{if
$l=o(\sqrt{n/m})$} \\
\bar{c}_2 \sqrt{nm}\log n &\textrm{if $l=\Omega(\sqrt{n/m})$,}
\end{array}\right. \nonumber \\&& \le \bar{c}_2
m\min\left\{l,\sqrt{\frac{n}{m}}\right\}\log n
\end{eqnarray}
where $\bar{c}_2>0$ is some constant independent of $n$. Next, the
third term $T_{n}^{(3)}$ of (\ref{EQ:thcut1}) will be shown by using
(\ref{EQ:Ptransfer}), (\ref{EQ:PdLi5}) and Lemma
\ref{LEM:EXdistance}. If $l=o(\sqrt{n/m})$, which corresponds to
operating regimes A and B shown in Fig. \ref{FIG:scaling_extended},
then $T_{n}^{(3)}$ is given by
\begin{eqnarray}
T_{n}^{(3)}= \left\{\begin{array}{lll} O(n^{2-\alpha/2+\epsilon})
&\textrm{if $2<\alpha<3$}
\\
O(n^{1/2+\epsilon}) &\textrm{if $\alpha\ge3$.}
\end{array}\right. \nonumber
\end{eqnarray}
Hence, under this network condition,
\begin{eqnarray}
T_{n}\le
c_5n^{\epsilon}\max\left\{ml,\sqrt{n},n^{2-\alpha/2}\right\}
\nonumber
\end{eqnarray}
for some constant $c_5>0$ independent of $n$, which is upper-bounded
by the RHS of (\ref{EQ:excutset}). Now we focus on the case for
$l=\Omega(\sqrt{n/m})$ (regimes C and D in Fig.
\ref{FIG:scaling_extended}). In this case, $T_{n}^{(3)}$ is
upper-bounded by
\begin{eqnarray} \label{EQ:TnL13power}
T_{n}^{(3)}\!\!\!\!\!\!\!\!&&\le\left\{\begin{array}{lll}
\bar{c}_3n^{\epsilon}\left(n^{2-\alpha/2}(\log n)^2+
nl\left(\frac{m}{n}\right)^{\alpha/2}\log n\right) &\textrm{if
$2<\alpha<3$} \\
\bar{c}_3n^{\epsilon}\left(\sqrt{n}(\log
n)^3+ml\sqrt{\frac{m}{n}}(\log n)^2\right) &\textrm{if
$\alpha=3$} \\
\bar{c}_3n^{\epsilon}\left(\sqrt{n}(\log
n)^2+\frac{n}{\sqrt{l}}\left(\frac{ml}{n}\right)^{\alpha/2}\log
n\right) &\textrm{if $\alpha>3$}
\end{array}\right. \nonumber\\ && \le\left\{\begin{array}{lll}
\bar{c}_3n^{\epsilon_2}\max\left\{n^{2-\alpha/2},
nl\left(\frac{m}{n}\right)^{\alpha/2}\right\} &\textrm{if
$2<\alpha<3$} \\
\bar{c}_3n^{\epsilon_2}\max\left\{\sqrt{n},\frac{n}{\sqrt{l}}\left(\frac{ml}{n}\right)^{\alpha/2}\right\}
&\textrm{if $\alpha\ge3$}
\end{array}\right.
\end{eqnarray}
for some constant $\bar{c}_3>0$ and an arbitrarily small constant
$\epsilon_2>\epsilon>0$. From (\ref{EQ:TnL11}), (\ref{EQ:TnL12}),
and (\ref{EQ:TnL13power}), we thus get the following result:
\begin{eqnarray}
T_{n}\!\!\!\!\!\!\!&&\le\left\{\begin{array}{lll}
\bar{c}_4n^{\epsilon}\max\left\{\sqrt{nm},n^{2-\alpha/2},nl\left(\frac{m}{n}\right)^{\alpha/2}\right\}
&\textrm{if
$2<\alpha<3$} \\
\bar{c}_4n^{\epsilon}\max\left\{\sqrt{nm},\frac{n}{\sqrt{l}}\left(\frac{ml}{n}\right)^{\alpha/2}\right\}
&\textrm{if $\alpha\ge3$}
\end{array}\right. \nonumber\\ && \le
\bar{c}_4n^{\epsilon}\max\left\{\sqrt{nm},n^{2-\alpha/2},ml\left(\frac{m}{n}\right)^{\alpha/2-1}\right\},
\nonumber
\end{eqnarray}
where the first and second inequalities hold since
$\sqrt{nm}=\Omega(\sqrt{n})$ and
$\sqrt{nm}=\Omega\left(\frac{n}{\sqrt{l}}\left(\frac{ml}{n}\right)^{\alpha/2}\right)$,
respectively, which results in (\ref{EQ:excutset}). This completes
the proof of Theorem~\ref{THM:EXupper}.

Now we would like to examine in detail the amount of information
transfer by each separated destination set.

\begin{remark}
The information transfer by the BS antennas on the left half, i.e.,
the destination set $D_L \setminus \bar{D}_r$, becomes dominant
under operating regimes B--D (especially at the high path-loss
attenuation regimes) in Fig.~\ref{FIG:scaling_extended}. More
specifically, compared to the pure network case with no BSs, as $m$
and $l$ increases (i.e., regimes B--D), enough DoF gain is obtained
by exploiting multiple antennas at each BS, while the power gain is
provided since all the BSs are connected by the wired BS-to-BS
links. In addition, note that the first to fourth terms in
(\ref{EQ:excutset}) represent the amount of information transferred
to the destination sets $D_{L}\setminus(D_{L}^{(2)} \cup
\bar{D}_r)$, $D_{L}^{(2)}$, $D_{L}^{(1)}$, and $\bar{D}_r\setminus
D_{L}^{(1)}$, and can be achieved by the ISH, IMH, MH, HC schemes,
respectively.
\end{remark}


\section{Conclusion} \label{SEC:conclusion}

The paper has analyzed the benefits of infrastructure support for
generalized hybrid networks. Provided the number $m$ of BSs and the
number $l$ of antennas at each BS scale at arbitrary rates relative
to the number $n$ of wireless nodes, the capacity scaling laws were
derived as a function of these scaling parameters. Specifically, two
routing protocols using BSs were proposed, and their achievable rate
scalings were derived and compared with those of the two
conventional schemes MH and HC in both dense and extended networks.
Furthermore, to show the optimality of the achievability results,
new information-theoretic upper bounds were derived. In both dense
and extended networks, it was shown that our achievable schemes are
order-optimal for all the operating regimes.


\appendix

\section{Appendix}

\subsection{Achievable Throughput With Respect to Operating Regimes} \label{PF:extended}

Let $e_{\text{ISH}}$, $e_{\text{IMH}}$, $_{\text{MH}}$, and
$e_{\text{HC}}$ denote the scaling exponents for the achievable
throughput of the ISH, IMH, MH, and HC protocols, respectively. The
scaling exponents among the above schemes are compared according to
operating regimes A--D shown in Fig. \ref{FIG:scaling_extended}
($\epsilon$ is omitted for notational convenience). From the result
of Theorem~\ref{THM:exlower}, note that $e_{\text{ISH}}$,
$e_{\text{MH}}$, and $e_{\text{HC}}$ are given by
$1+\gamma-\frac{(1-\beta)\alpha}{2}$, $\frac{1}{2}$, and
$2-\frac{\alpha}{2}$, respectively, regardless of operating regimes.
\begin{enumerate}
\item Regime A ($0\le\beta+\gamma<\frac{1}{2}$): $e_{\text{IMH}}=\beta+\gamma$ is obtained. Since
$e_{\text{MH}}>e_{\text{IMH}}>e_{\text{ISH}}$, pure ad hoc
transmissions with no BSs outperform the ISH and IMH protocols.
Hence, the results in Regime A of TABLE \ref{T:table1} are obtained.

\item Regime B ($\beta+\gamma\ge\frac{1}{2}$ and $\beta+2\gamma<1$): $e_{\text{IMH}}$ is the same as that under Regime A. Since
$e_{\text{IMH}}>e_{\text{ISH}}$ and $e_{\text{IMH}}\ge
e_{\text{MH}}$, the IMH always outperforms the ISH and the MH.
Hence, it is found that the HC scheme has the largest scaling
exponent under $2<\alpha<4-2\beta-2\gamma$, but if
$\alpha\ge4-2\beta-2\gamma$ the IMH protocol becomes the best.

\item Regime C ($\beta+2\gamma\ge1$ and
$\gamma<\frac{1}{2}(\beta^2-3\beta+2)$): Remark that
$e_{\text{IMH}}=\frac{1+\beta}{2}$ and $e_{\text{IMH}} \ge
e_{\text{MH}}$. Then, the following inequalities with respect to the
path-loss exponent $\alpha$ are found:
$e_{\text{ISH}}>e_{\text{IMH}}$ for
$2<\alpha<1+\frac{2\gamma}{1-\beta}$ and $e_{\text{ISH}} \le
e_{\text{IMH}}$ for $\alpha\ge1+\frac{2\gamma}{1-\beta}$;
$e_{\text{HC}}>e_{\text{IMH}}$ for $2<\alpha<3-\beta$ and
$e_{\text{HC}}\le e_{\text{IMH}}$ for $\alpha\ge 3-\beta$; and
$e_{\text{HC}}>e_{\text{ISH}}$ for
$2<\alpha<\frac{2(1-\gamma)}{\beta}$ and $e_{\text{HC}}\le
e_{\text{ISH}}$ for $\alpha\ge\frac{2(1-\gamma)}{\beta}$. The best
scheme thus depends on the comparison among
$1+\frac{2\gamma}{1-\beta}$, $3-\beta$, and
$\frac{2(1-\gamma)}{\beta}$. Note that
$3-\beta<\frac{2(1-\gamma)}{\beta}$ and
$3-\beta>1+\frac{2\gamma}{1-\beta}$ always hold under Regime C.
Finally, the best achievable schemes with respect to $\alpha$ are
obtained and are shown in Fig. \ref{FIG:regimeC}.

\item Regime D ($\beta+\gamma<1$ and
$\gamma\ge\frac{1}{2}(\beta^2-3\beta+2)$): The same scaling
exponents for our four protocols are the same as those under Regime
C. The result is obtained by comparing $1+\frac{2\gamma}{1-\beta}$,
$3-\beta$, and $\frac{2(1-\gamma)}{\beta}$ under Regime D. The
following two inequalities $3-\beta\ge\frac{2(1-\gamma)}{\beta}$ and
$3-\beta\le1+\frac{2\gamma}{1-\beta}$ are satisfied, and the best
achievable schemes with respect to $\alpha$ are obtained and shown
in Fig. \ref{FIG:regimeD}.
\end{enumerate}
This coincides with the result shown in TABLE \ref{T:table1}.


\subsection{Proof of Lemma~\ref{LEM:intISH}} \label{PF:intISH}

First consider the uplink case. There are $8k$ interfering cells,
each of which includes $\Theta(n/m)$ nodes whp, in the $k$-th layer
$l_{k}$ of the network as illustrated in Fig. \ref{FIG:layer}. Let
$d_k$ denote the Euclidean distance between a given BS antenna and
any node in $l_{k}$, which is a random variable. Since $d_k$ scales
as $\Theta(k\sqrt{n/m})$, there exists $c_7>c_6>0$ with constants
$c_6$ and $c_7$ independent of $n$, such that
$d_k=c_{8}k\sqrt{n/m}$, where all $c_8$ lies in the interval $[c_6,
c_7]$. Hence, the total interference power $P_I^u$ at each BS
antenna from simultaneously transmitting nodes is upper-bounded by
\begin{eqnarray}
P_I^u\!\!\!\!\!\!\!&&\le\sum_{k=1}^{\infty}\frac{P}{
(m/n)^{\alpha/2-1}}(8k)\frac{n}{m}\left(\frac{m}{(c_6k)^2n}\right)^{\alpha/2}
\nonumber\\
&&=\frac{8P}{c_6^{\alpha}}\sum_{k=1}^{\infty}\frac{1}{k^{\alpha-1}}
\nonumber
\\ &&\le c_9, \nonumber
\end{eqnarray}
where $c_9>0$ is some constant independent of $n$. Now let us
consider the interference in the downlink. The interfering signal
received by node $i$, which is in the cell covered by BS $s$, from
the simultaneously operating BSs $s'\in\{1,\cdots,m\}\setminus
\{s\}$ is given by
\begin{equation}
\sum_{s'\in\{1,\cdots,m\}\setminus \{s\}}{\bf
h}_{is'}^{d}\left(\sum_{j=1}^{n/m}{\bf u}_j^{s'}x_j^{s'}\right),
\nonumber
\end{equation}
where ${\bf u}_j^{s'}$ denotes the $j$-th transmit precoding vector
at BS $s'$ normalized so that its $L_2$-norm is unity and $x_j^{s'}$
is the $j$-th transmit packet at BS $s'$. Since ${\bf u}_j^{s'}$ is
represented by a function of the downlink channel coefficients
between BS $s'$ and nodes communicating with BS $s'$, the terms
$\left[{\bf h}_{is'}^{d}\right]_k\!\cdot\! \left[\sum_j{\bf
u}_j^{s'}x_j^{s'}\right]_k$ are independent for all
$k\in\{1,\cdots,n/m\}$ and $s'\in\{1,\cdots,m\}\setminus \{s\}$.
Using the fact above and the layering technique similar to the
uplink case, an upper bound of the average total interference power
$P_I^d$ at each node in the downlink is obtained as the following:
\begin{equation}
P_I^d\le\sum_{k=1}^{\infty}\frac{(n/m)P}{(m/n)^{\alpha/2-1}}(8k)\left(\frac{m}{(c_{6}k)^2n}\right)^{\alpha/2}
\le c_9', \nonumber
\end{equation}
where $c_9'>0$ is some constant independent of $n$.


\subsection{Proof of Lemma \ref{LEM:numBSlink}} \label{PF:numBSlink}

Let $X_i$ denote the number of sources in the $i$-th cell and
$\mathcal{E}_{x}$ denote the event that $X_i$ is between
$((1-\delta_0)n/m, (1+\delta_0)n/m)$ for all $i\in\{1,\cdots,m\}$,
where $0<\delta_0<1$ is some constant independent of $n$. Then, we
have
\begin{eqnarray} \label{EQ:PrXa}
&&\Pr\left\{X_{ki}< a \mbox{ for all } i,k\in\{1,\cdots,m\}\right\}
\nonumber\\ &&\ge \Pr\left\{\mathcal{E}_x\right\}\Pr\left\{X_{ki}< a
\mbox{ for all } i,k|\mathcal{E}_x\right\} \nonumber\\ &&\ge
\Pr\left\{\mathcal{E}_x\right\}\left(1-m^2\Pr\left\{\sum_{j=1}^{(1+\delta_0)n/m}B_j\ge
a\right\}\right),
\end{eqnarray}
where $\sum_j B_j$ is the sum of $(1+\delta_0)n/m$ independent and
identically distributed (i.i.d.) Bernoulli random variables with
probability
\begin{equation}
\Pr\left\{B_j=1\right\}=\frac{1}{m}. \nonumber
\end{equation}
Here, the second inequality holds since the union bound is applied
over all $i,k\in\{1,\cdots,m\}$. We first consider the case where
$n/m=\omega(m)$, i.e., $0\le\beta<1/2$. By setting
$a=(1+\delta_0)^2n/m^2$, we then get
\begin{eqnarray}
\Pr\left\{\sum_{j=1}^{(1+\delta_0)n/m}B_j\ge (1+\delta_0)^2\frac{n}{m^2}\right\}\!\!\!\!\!\!\!&&=\Pr\left\{e^{s\sum_{j=1}^{(1+\delta_0)n/m}B_j}\ge e^{s(1+\delta_0)^2n/m^2}\right\} \nonumber\\
&& \le e^{-(1+\delta_0)n/m^2\left(s(1+\delta_0)-e^s+1\right)},
\label{EQ:PrB}
\end{eqnarray}
which is derived from the steps similar to the proof of Lemma 4.1
in~\cite{OzgurLevequeTse:07}, where the first inequality comes from
an application of Chebyshev's inequality. Hence, using
(\ref{EQ:numprob}), (\ref{EQ:PrXa}), and (\ref{EQ:PrB}) yields
\begin{eqnarray}
&&\Pr\left\{X_{ki}< (1+\delta_0)^2\frac{n}{m^2} \mbox{ for all }
i,k\in\{1,\cdots,m\}\right\} \nonumber\\
&&\ge \left(1-n^\beta e^{-\Delta(\delta_0)n^{1-\beta}}\right)\left(1-m^2e^{-(1+\delta_0)\Delta(\delta_0)n/m^2}\right) \nonumber\\
&&=\left(1-n^\beta
e^{-\Delta(\delta_0)n^{1-\beta}}\right)\left(1-e^{2\beta\ln
n-(1+\delta_0)\Delta(\delta_0)n^{1-2\beta}}\right), \nonumber
\end{eqnarray}
where $\Delta(\delta_0)=(1+\delta_0)\ln(1+\delta_0)-\delta_0$, by
choosing $s=\ln(1+\delta_0)$, which converges to one as $n$ goes to
infinity. When $n/m=O(m)$, i.e., $1/2\le\beta<1$, setting $a=\ln n$
and $s=(2+\delta_0)\beta$ and following the approach similar to the
first case, we obtain
\begin{eqnarray}
&&\Pr\left\{X_{ki}< \ln n \mbox{ for all }
i,k\in\{1,\cdots,m\}\right\} \nonumber\\
&&\ge \left(1-n^\beta e^{-\Delta(\delta_0)n^{1-\beta}}\right)\left(1-m^2e^{-(1+\delta_0)\left(1-e^{(2+\delta_0)\beta}\right)n/m^2-(2+\delta_0)\beta\ln n}\right) \nonumber\\
&&=\left(1-n^\beta
e^{-\Delta(\delta_0)n^{1-\beta}}\right)\left(1-e^{-\delta_0\beta\ln
n-(1+\delta_0)\left(1-e^{(2+\delta_0)\beta}\right)n^{1-2\beta}}\right),
\nonumber
\end{eqnarray}
which converges to one as $n$ goes to infinity. This completes the
proof.


\subsection{Proof of Lemma \ref{LEM:distance}} \label{PF:distance}

This result can be obtained by slightly modifying the asymptotic
analysis in~\cite{OzgurLevequeTse:07,ElGamalMammenPrabhakarShah:06}.
The minimum node-to-node distance is easily derived by following the
same approach as that in~\cite{OzgurLevequeTse:07} and is proved to
scale at least as $1/n^{1+\epsilon_1}$ with probability
$1-\Theta(1/n^{2\epsilon_1})$. We now focus on how the distance
between a node and an antenna on the BS boundary scales. Consider a
circle of radius $1/n^{1+\epsilon_1}$ around one specific antenna on
the BS boundary. Note that there are no other antennas inside the
circle since the per-antenna distance is greater than
$1/n^{1+\epsilon_1}$. Let $\mathcal{E}_d$ denote the event that $n$
nodes are located outside the circle given by the antenna. Then, we
have
\begin{equation}
P\{\mathcal{E}_d^{\mathcal{C}}\}\le
1-\left(1-\frac{c_{10}\pi}{n^{2+2\epsilon_1}}\right)^{n}, \nonumber
\end{equation}
where $0<c_{10}<1$ is some constant independent of $n$. Hence, by
the union bound, the probability that the event $\mathcal{E}_d$ is
satisfied for all the BS antennas is lower-bounded by
\begin{eqnarray}
1-ml
P\{\mathcal{E}_d^{\mathcal{C}}\}\!\!\!\!\!\!\!&&\ge1-ml\left(1-\left(1-\frac{c_{10}\pi}{n^{2+2\epsilon_1}}\right)^{n}\right)
\nonumber \\ &&\ge
1-n\left(1-\left(1-\frac{c_{10}\pi}{n^{2+2\epsilon_1}}\right)^{n}\right),
\nonumber
\end{eqnarray}
where the second inequality holds since $ml=O(n)$, which tends to
one as $n$ goes to infinity. This completes the proof.


\subsection{Proof of Lemma \ref{LEM:Hsing}} \label{PF:Hsing}

The size of matrix ${\bf F}_{L}^{(3)}$ is $\Theta(n)\times
\Theta(n)$ since $ml=O(n)$. Thus, the analysis essentially follows
the argument in~\cite{OzgurLevequeTse:07} with a slight modification
(see Appendix III in~\cite{OzgurLevequeTse:07} for more precise
description). Consider the network transformation resulting in a
regular network with at most $\log n$ nodes at each square vertex
except for the area covered by BSs. Then, the same node displacement
as shown in~\cite{OzgurLevequeTse:07} is performed, which will
decrease the Euclidean distance between source and destination
nodes. For convenience, the source node positions are indexed in the
resulting regular network. It is thus assumed that the source nodes
under the cut are located at positions $(-i_x+1,i_y)$ where
$i_x,i_y=1,\cdots \sqrt{n}$. In the following, $\sum_{k\in
D_{L}^{(3)}}\left|\left[{\bf F}_{L}^{(3)}\right]_{ki}\right|^2$ and
an upper bound for $\sum_{i\in S_{L}}\left|\left[{\bf
F}_{L}^{(3)}\right]_{ki}\right|^2$ are derived:
\begin{eqnarray}
\sum_{k\in D_{L}^{(3)}}\left|\left[{\bf
F}_{L}^{(3)}\right]_{ki}\right|^2=1 \nonumber
\end{eqnarray}
and
\begin{eqnarray}
\sum_{i\in S_{L}}\left|\left[{\bf
F}_{L}^{(3)}\right]_{ki}\right|^2\!\!\!\!\!\!\!\!&&=\underset{i\in
S_{L}}\sum \left|\frac{1}{\sqrt{d_{L,i}^{(3)}}}\left[{\bf
H}_{L}^{(3)}\right]_{ki}\right|^2 \nonumber\\
&& = \left\{\begin{array}{lll} \underset{i\in S_{L}}\sum
\frac{r_{ki}^{-\alpha}}{d_{L,i}^{(3)}} &\textrm{if
$k\in \bar{D}_r\setminus D_{L}^{(1)}$} \\
\underset{i\in S_{L}}\sum
\frac{r_{si,t}^{u~\!-\alpha}}{d_{L,i}^{(3)}} &\textrm{if $k\in \{t:
t\in[1,l]~\text{for}~s\in B_l\}$}
\end{array}\right. \nonumber\\
&& \le \left\{\begin{array}{lll} \underset{i\in S_{L}}\sum
\frac{r_{ki}^{-\alpha}}{\underset{k\in\bar{D}_r \setminus
D_{L}^{(1)}}\sum r_{ki}^{-\alpha}} &\textrm{if
$k\in \bar{D}_r\setminus D_{L}^{(1)}$} \\
\underset{i\in S_{L}}\sum
\frac{r_{si,t}^{u~\!-\alpha}}{\underset{k\in\bar{D}_r \setminus
D_{L}^{(1)}}\sum r_{ki}^{-\alpha}} &\textrm{if $k\in \{t:
t\in[1,l]~\text{for}~s\in B_l\}$}
\end{array}\right. \nonumber\\
&& \le \left\{\begin{array}{lll} c_{11}\log n\underset{i\in
S_{L}}\sum x_i^{\alpha-2}r_{ki}^{-\alpha} &\textrm{if
$k\in \bar{D}_r\setminus D_{L}^{(1)}$} \\
c_{11}\log n\underset{i\in S_{L}}\sum
x_i^{\alpha-2}r_{si,t}^{u~\!-\alpha} &\textrm{if $k\in \{t:
t\in[1,l]~\text{for}~s\in B_l\}$}
\end{array}\right. \nonumber\\
&& \le \left\{\begin{array}{lll} c_{11}\log n\underset{i\in
S_{L}}\sum r_{ki}^{-2} &\textrm{if
$k\in \bar{D}_r\setminus D_{L}^{(1)}$} \\
c_{11}\log n\underset{i\in S_{L}}\sum r_{si,t}^{u~\!-2} &\textrm{if
$k\in \{t: t\in[1,l]~\text{for}~s\in B_l\}$}
\end{array}\right. \nonumber\\
&&\le c_{11}(\log n)^2 \sum_{i_x,i_y=1}^{\sqrt{n}}
\frac{1}{i_x^2+i_y^2} \nonumber\\
&&\le c_{12} (\log n)^3, \nonumber
\end{eqnarray}
where $\bar{D}_r$ is the set of nodes including BS antennas on the
right half, $B_l$ is the set of BSs in the left half network,
$c_{11}$ and $c_{12}$ are some positive constants independent of
$n$, and $x_i$ denotes the $x$-coordinate of node $i\in S_{L}$ for
our random network ($x_i=1,\cdots,\sqrt{n}$). Here, the second and
fifth inequalities hold since
\begin{equation}
\underset{k\in\bar{D}_r \setminus D_{L}^{(1)}}\sum
r_{ki}^{-\alpha}\ge\frac{x_i^{2-\alpha}}{c_{11}\log n} \nonumber
\end{equation}
and
\begin{equation}
\sum_{i_x,i_y=1}^{\sqrt{n}} \frac{1}{i_x^2+i_y^2}=O(\log n),
\nonumber
\end{equation}
respectively (see Appendix III in~\cite{OzgurLevequeTse:07} for the
detailed derivation). The fourth inequality comes from the result of
Lemma~\ref{LEM:randomlog}. Hence, it is proved that both scaling
results are the same as the random network case shown
in~\cite{OzgurLevequeTse:07}.

Now it is possible to prove the inequality in (\ref{EQ:ineqHLp}) by
following the same line as that in Appendix III
of~\cite{OzgurLevequeTse:07}, which results in
\begin{equation}
E\left[\tr\left(\left({{\bf F}_{L}^{(3)}} ^{\dagger}{{\bf
F}_{L}^{(3)}}\right)^q\right)\right]\le C_q n\left(c_{13}\log
n\right)^{3q}, \label{EQ:HLpineq} \nonumber
\end{equation}
where $C_q=\frac{(2q)!}{q!(q+1)!}$ is the Catalan number for any $q$
and $c_{13}>0$ is some constant independent of $n$. Then, from the
property $\|{\bf F}_{L}^{(3)}\|_2^2=\lim_{q\rightarrow\infty}
\tr(({{\bf F}_{L}^{(3)}}^{\dagger}{\bf F}_{L}^{(3)})^q)^{1/q}$ (see
\cite{HornJohnson:99}), the expectation of the term $\|{\bf
F}_{L}^{(3)}\|_2^2$ is finally given by (\ref{EQ:ineqHLp}), which
completes the proof.


\subsection{Proof of Lemma \ref{LEM:EXdistance}} \label{PF:EXdistance}

When $l=o(\sqrt{n/m})$, there is no destination in $D_{L}^{(5)}$,
and thus $\sum_{i\in S_{L}} Pd_{L,i}^{(5)}$ becomes zero. Hence, the
case for $l=\Omega(\sqrt{n/m})$ is the focus from now on. By the
same argument as shown in the derivation of $\sum_{i\in
S_{L}}Pd_{L,i}^{(4)}$, the network area is divided into $n$ squares
of unit area. Then, by Lemma~\ref{LEM:randomlog}, the power transfer
under our random network can be upper-bounded by that under a
regular network with at most $\log n$ nodes at each square except
for the area covered by BSs. As illustrated in Fig.
\ref{FIG:regularregime}, the nodes in each square are moved together
onto one vertex of the corresponding square. The node displacement
is performed in a sense of decreasing the Euclidean distance between
node $i\in S_{L}$ and the antennas of the corresponding BS, thereby
providing an upper bound for $d_{L,i}^{(5)}$. Layers of each cell
are then introduced, as shown in Fig. \ref{FIG:regularregime}, where
there exist $8(\epsilon_0\sqrt{n/m}+k)$ vertices, each of which
includes $\log n$ nodes, in the $k$-th layer $l_{k}'$ of each cell.
The regular network described above can also be transformed into the
other, which contains antennas regularly placed at spacing
$\epsilon_0\sqrt{\frac{n\pi}{ml}}$ outside the shaded square for an
arbitrarily small $\epsilon_0>0$. Note that the shaded square of
size $2k\times 2k$ is drawn based on a source node in $l_{k}'$ at
its center (see Fig. \ref{FIG:regularregime}). The modification
yields an increase of the term $d_{L,i}^{(5)}$ by source $i$. When
$d_{L,i(k)}^{(5)}$ is defined as $d_{L,i}^{(5)}$ by node $i$ that
lies in $l_{k}'$, the following upper bound for $d_{L,i(k)}^{(5)}$
is obtained:
\begin{eqnarray}
d_{L,i(k)}^{(5)}\!\!\!\!\!\!\!&&\le
\sum_{k_x,k_y=\zeta}^{\infty}\frac{1}{\left(\left(\epsilon_0\sqrt{\frac{n\pi}{ml}}k_x)^2+(\epsilon_0\sqrt{\frac{n\pi}{ml}}k_y\right)^2\right)^{\alpha/2}}
\nonumber \\
&&=\left(\frac{ml}{n}\right)^{\alpha/2}\sum_{k_x,k_y=\zeta}^{\infty}\frac{{\eta}^{\alpha}}{(k_x^2+k_y^2)^{\alpha/2}}
\nonumber \\
&&\le\left(\frac{ml}{n}\right)^{\alpha/2}\sum_{k'=\zeta}^{\infty}\frac{8{\eta}^{\alpha}k'}{{k'}^{\alpha}}
\nonumber\\ &&\le
c_{14}\left(\frac{ml}{n}\right)^{\alpha/2}\left(\frac{1}{\zeta^{\alpha-1}}+\int_{\zeta}^{\infty}\frac{1}{x^{\alpha-1}}dx\right)
\nonumber \\
&&=
{c}_{14}\left(\frac{1}{\zeta}+\frac{1}{\alpha-2}\right)\left(\frac{ml}{n}\right)^{\alpha/2}\frac{1}{\zeta^{\alpha-2}},
\label{EQ:dL1ik} \nonumber
\end{eqnarray}
where $\zeta=1+\lfloor k\eta\rfloor$,
$\eta=\frac{1}{\epsilon_0}\sqrt{\frac{ml}{n\pi}}$, and $c_{14}$ is
some constant independent of $n$. Here, $\lfloor x \rfloor$ denotes
the greatest integer less than or equal to $x$. Hence,
$d_{L,i(k)}^{(5)}$ is given by
\begin{eqnarray}
d_{L,i(k)}^{(5)}=\left\{\begin{array}{ll}
O\left(\left(\frac{ml}{n}\right)^{\alpha/2}\right) &\textrm{if
$k=O\left(\sqrt{\frac{n}{ml}}\right)$} \\
O\left(k^{2-\alpha}\left(\frac{ml}{n}\right)\right) &\textrm{if
$k=\Omega\left(\sqrt{\frac{n}{ml}}\right)$,}\end{array}\right.
\nonumber
\end{eqnarray}
finally yielding
\begin{eqnarray}
\sum_{i\in S_{L}} Pd_{L,i}^{(5)}\!\!\!\!\!\!\!&&\le
P\frac{m}{2}\log n\sum_{k=1}^{\sqrt{n/m}} 8\left(\epsilon_0\sqrt{\frac{n}{m}}+k\right)d_{L,i(k)}^{(5)} \nonumber \\
&&\le c_{15}P\sqrt{nm}\log
n\left[\sum_{k=1}^{\sqrt{n/ml}-1}\left(\frac{ml}{n}\right)^{\alpha/2}+\sum_{k=\sqrt{n/ml}}^{\sqrt{n/m}}k^{2-\alpha}\left(\frac{ml}{n}\right)\right]
\nonumber\\ &&\le c_{15}P\sqrt{nm}\log
n\left[\left(\frac{ml}{n}\right)^{(\alpha-1)/2}+\left(\frac{ml}{n}\right)\left(\left(\frac{ml}{n}\right)^{\alpha/2-1}+\int_{\sqrt{n/ml}}^{\sqrt{n/m}}\frac{1}{k^{\alpha-2}}dx\right)\right]
\nonumber\\ &&\le \left\{\begin{array}{lll}
\frac{3c_{15}P}{3-\alpha}nl\left(\frac{m}{n}\right)^{\alpha/2}\log n
&\textrm{if
$2<\alpha<3$} \\
\frac{3c_{15}P}{2}ml\sqrt{\frac{m}{n}}(\log n)^2 &\textrm{if
$\alpha=3$} \\
\frac{3c_{15}P}{\alpha-3}\frac{n}{\sqrt{l}}\left(\frac{ml}{n}\right)^{\alpha/2}\log
n &\textrm{if $\alpha>3$,}\end{array}\right. \label{EQ:sumalpha}
\end{eqnarray}
where $c_{15}$ is some constant independent of $n$. Here, the first
inequality holds since there exist $8(\epsilon_0\sqrt{n/m}+k)$
vertices in $l_{k}'$ and at most $\log n$ nodes at each vertex.
Equation (\ref{EQ:sumalpha}) yields the result in (\ref{EQ:PL1sum}),
which completes the proof.



\newpage

\begin{figure}
  \begin{center}
  \scalebox{0.74}{\includegraphics{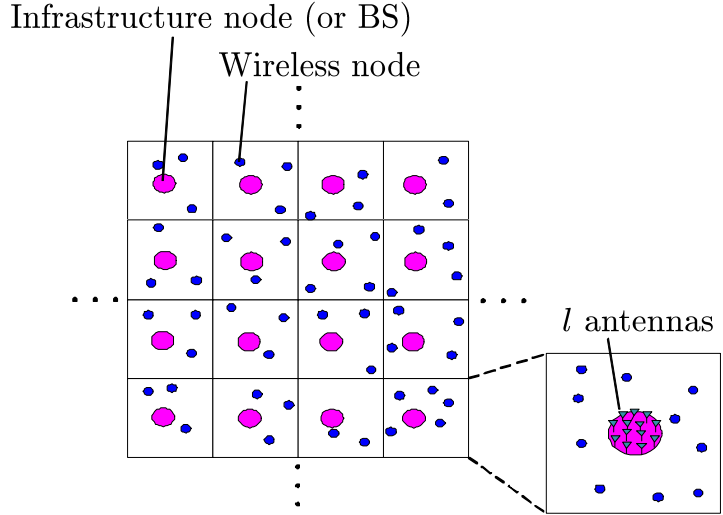}}
  \caption{The wireless ad hoc network with infrastructure support.}
  \label{FIG:infra}
  \end{center}
\end{figure}

\begin{figure}
  \begin{center}
  \scalebox{0.70}{\includegraphics{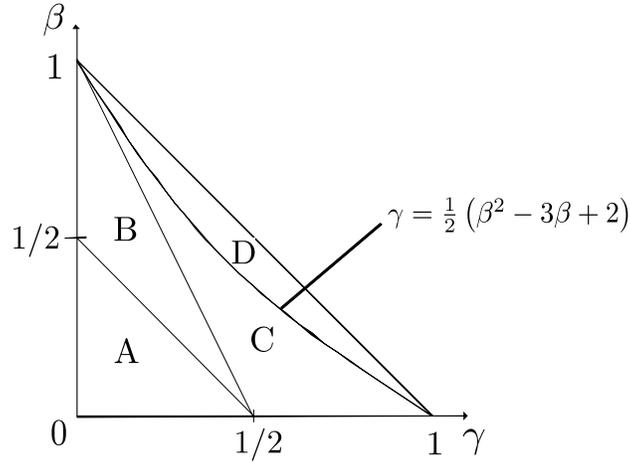}}
  \caption{Operating regimes on the achievable throughput scaling with respect to $\beta$ and $\gamma$.}
  \label{FIG:scaling_extended}
  \end{center}
\end{figure}

\begin{figure}
  \begin{center}
  \scalebox{0.79}{\includegraphics{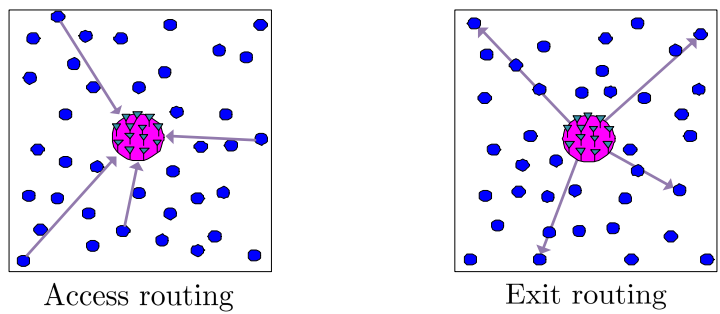}}
  \caption{The infrastructure-supported single-hop (ISH) protocol.}
  \label{FIG:ISH}
  \end{center}
\end{figure}

\begin{figure}
  \begin{center}
  \scalebox{0.79}{\includegraphics{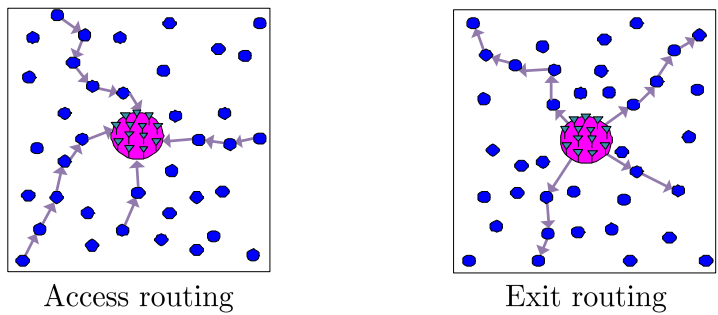}}
  \caption{The infrastructure-supported multi-hop (IMH) protocol.}
  \label{FIG:IMH}
  \end{center}
\end{figure}

\begin{figure}
  \begin{center}
  \scalebox{0.97}{\includegraphics{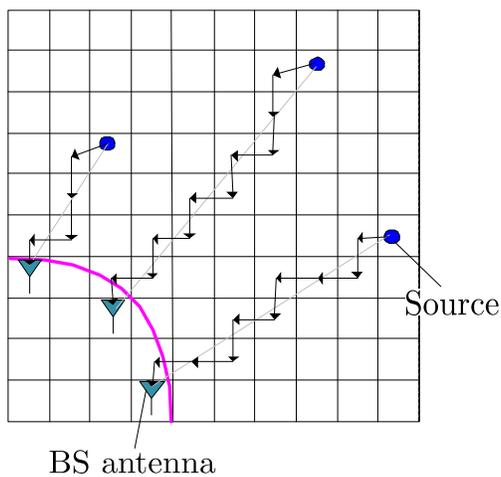}}
  \caption{The access routing in the IMH protocol.}
  \label{FIG:access}
  \end{center}
\end{figure}

\begin{figure}
  \begin{center}
  \scalebox{0.88}{\includegraphics{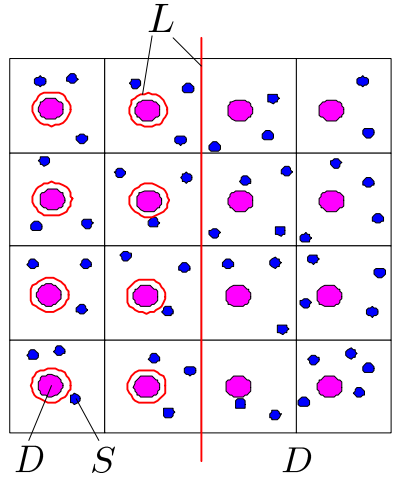}}
  \caption{The cut $L$ in the two-dimensional random network.}
  \label{FIG:cutset}
  \end{center}
\end{figure}

\begin{figure}
  \begin{center}
  \scalebox{0.97}{\includegraphics{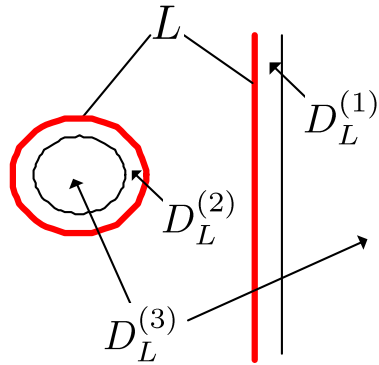}}
  \caption{The partition of destinations in the two-dimensional random network. To simplify the figure, one BS is shown in the left half network.}
  \label{FIG:Dcut}
  \end{center}
\end{figure}

\begin{figure}
  \begin{center}
  \subfigure[\!\!\!\!\!\!\!\!\!\!\!]{%
  \scalebox{0.69}{\includegraphics{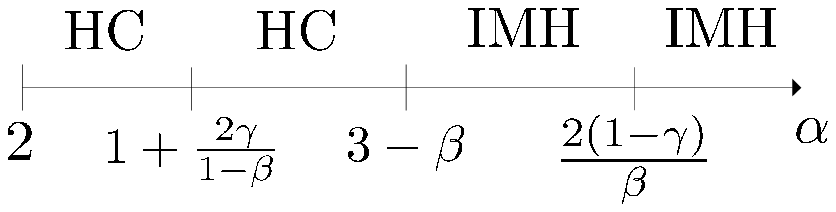}}\label{FIG:regimeC}}
  \vspace{.0cm}
  \subfigure[\!\!\!\!\!\!\!\!\!\!\!]{%
  \scalebox{0.72}{\includegraphics{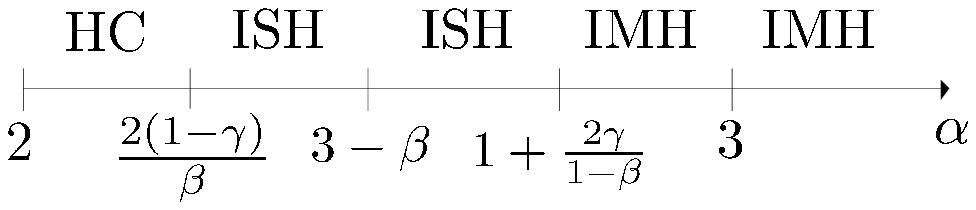}}\label{FIG:regimeD}}
  \caption{The best achievable schemes with respect to $\alpha$. (a) The Regime C. (b) The Regime D.}
  \label{FIG:regimeCD}
  \end{center}
\end{figure}

\begin{figure}
  \begin{center}
  \scalebox{0.61}{\includegraphics{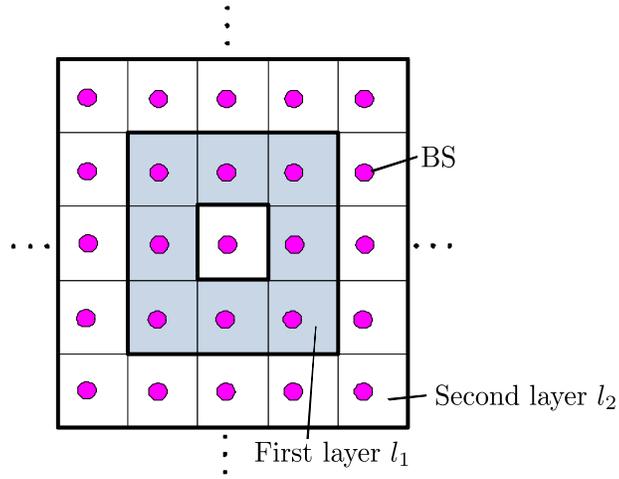}}
  \caption{Grouping of interfering cells. The first layer $l_{1}$ of the network represents the outer 8 shaded cells.}
  \label{FIG:layer}
  \end{center}
\end{figure}

\begin{figure}
  \begin{center}
  \scalebox{0.58}{\includegraphics{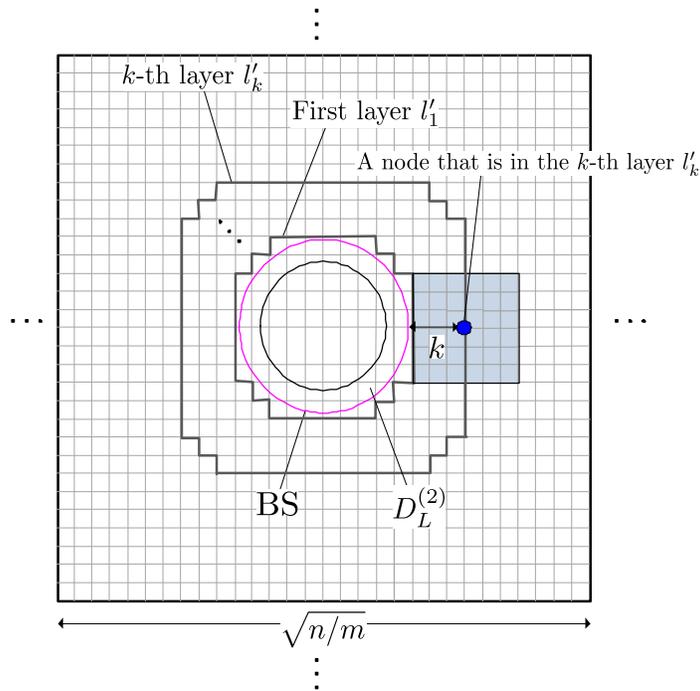}}
  \caption{The displacement of the nodes to square vertices. The antennas are regularly placed at spacing $\sqrt{\frac{n}{ml}}$ outside the shaded square.}
  \label{FIG:regularregime}
  \end{center}
\end{figure}

\begin{table}[t!]
\renewcommand{\arraystretch}{0.7}
\begin{center}
\caption{Achievable rates for an extended network with
infrastructure.}
\begin{tabular}{c|c|c|c} \hline %
Regime           & Condition                    & Scheme                   & $e(\alpha,\beta,\gamma)$     \\\hline \hline%
\vspace{-0.1cm}              &   $2<\alpha<3$ & HC  &   $2-\frac{\alpha}{2}$ \\
\vspace{-0.1cm}A                &   & \\
                 &   $\alpha\ge3$ & MH   &   $\frac{1}{2}$ \\ \hline %
\vspace{-0.1cm}              &   $2<\alpha<4-2\beta-2\gamma$ & HC  &   $2-\frac{\alpha}{2}$ \\
\vspace{-0.1cm}B                &   & \\
                 &   $\alpha\ge4-2\beta-2\gamma$ & IMH   &   $\beta+\gamma$ \\ \hline %
\vspace{-0.1cm}              &   $2<\alpha<3-\beta$ & HC  &   $2-\frac{\alpha}{2}$ \\
\vspace{-0.1cm}C                &   & \\
                 &   $\alpha\ge3-\beta$ & IMH   &   $\frac{1+\beta}{2}$ \\ \hline %
                &   $2<\alpha<\frac{2(1-\gamma)}{\beta}$ & HC    &   $2-\frac{\alpha}{2}$ \\
D                &   $\frac{2(1-\gamma)}{\beta}\le\alpha<1+\frac{2\gamma}{1-\beta}$ & ISH &  $1+\gamma-\frac{\alpha(1-\beta)}{2}$  \\
                  &   $\alpha\ge1+\frac{2\gamma}{1-\beta}$ & IMH  & $\frac{1+\beta}{2}$  \\ \hline %
\end{tabular}
\label{T:table1}
\end{center}
\end{table}

\end{document}